\RequirePackage{fix-cm}

\documentclass[onefignum,onetabnum]{siamart171218}

\ifpdf
\hypersetup{
  pdftitle={Global analysis of a simplified model of anaerobic digestion and
a new result for the chemostat},
  pdfauthor={T. Meadows, M. Weedermann, G.S.K. Wolkowicz}
}
\fi

\usepackage{amsmath,amsfonts}
\usepackage{fancyhdr}
\usepackage{enumitem}
\usepackage{graphicx}
\usepackage{color}
\usepackage{subfigure}
\usepackage{bibentry}
\usepackage{tikz}
\usepackage{ulem}
\usepackage{cite}
\usepackage{array,booktabs}

\usetikzlibrary{shapes.geometric, arrows}
\newsiamthm{prop}{Proposition}
\newsiamthm{thm}{Theorem}
\newsiamthm{cor}{Corollary}
\newsiamthm{lem}{Lemma}
\newsiamthm{defn}{Definition}
\newsiamthm{rmk}{Remark}
\def\proto#1{\mu_{2,\text{#1}}}

\title{Global analysis of a simplified model of anaerobic digestion and
a new result for the chemostat \thanks{Submitted to the editors \today}}
\author{T. Meadows\thanks{McMaster University \email{meadowta@mcmaster.ca}}
	 \and M. Weedermann \thanks{Dominican University \email{mweederm@dom.edu}}
	 \and G.S.K. Wolkowicz \thanks{McMaster University \email{wolkowic@math.mcmaster.ca} \funding{Natural Sciences and Engineering Discovery Grant \# 9358 and Accelerator supplement.}}
	}
\headers{T. Meadows, M. Weedermann, G.S.K. Wolkowicz}{Global analysis of a simplified model of anaerobic digestion}

\tikzstyle{box} = [rectangle, rounded corners, minimum width=2.4cm, minimum height=1cm,text centered, text width=2.5cm, draw=black]
\tikzstyle{arrow} = [thick, ->,>=stealth]
\tikzstyle{ball} = [circle, minimum height = 0.5cm, text centered, draw =black]

\begin{document}

\maketitle
\begin{abstract}
A. Bornh\"oft, R. Hanke-Rauschenbach, and K. Sundmacher, [ \textit{Nonlinear Dyn.}, 73 (2013), pp. 535--549] introduced a qualitative simplification to the
	ADM1 model for anaerobic digestion.
     We obtain global results for this model by
	first analyzing  the limiting system, a
	model of  single species growth in the
	chemostat    in which the response function is non-monotone and
	the species decay rate is included. Using a Lyapunov function argument
	and the theory of asymptotically autonomous systems,
	we prove that even in the
	parameter regime where there is bistability, no periodic orbits
	exist and every solution converges to one of the equilibrium
	points.   We then describe two algorithms for stochastically
	perturbing the parameters of the model. Simulations done with
	these two algorithms are compared with   simulations
	done  using the Gillespie
	and tau-leaping algorithms.   They
	  illustrate the
	severe impact environmental factors may have on anaerobic
	digestion in the transient phase.
\end{abstract}
\section{Introduction}
 Anaerobic digestion is a biochemical process where microorganisms or  multicellular organisms break down   organic material in the absence of oxygen.  Anaerobic digestion is an important part of many industrial practices, including the treatment of wastewater and the production of biogas. The role of anaerobic digestion in such applications has been an active area of recent research \cite{ADM1,Benyahia:2012,Bernard:2001,Bornhoft:2013,HMH:2010,HB:2008,Jeya:1997,Lokshina:2001,Shen:2007,SoRWE:2005}. This paper focuses on a particular model of anaerobic digestion
 in biogas production.

The foundation of previous work on the mathematical analysis of the
production of biogas is the \emph{Anaerobic Digestion Model 1} (ADM1)
\cite{ADM1} introduced in 2002. If implemented as a system of
differential equations, this model has 32 state variables, including
seven different species of microorganisms.  Understandably, anything other than numerical analysis has not been feasible.

In an effort to formally analyze the system, several groups \cite{Bornhoft:2013,HMH:2010,HB:2008,WeeSeoWolk:2013} have studied various subsystems  of ADM1. Recently, Weedermann et al. \cite{WeeSeoWolk:2013,WeeWolSas:2015}  combined two previous models \cite{HMH:2010,HB:2008} to get a reasonably complete picture using only eight state variables.
Due to the inclusion of two pathways to biogas production in \cite{WeeSeoWolk:2013} and because the model captures the ADM1's sensitivity to the accumulation of acetic acid,
\cite{WeeSeoWolk:2013} illustrates some of the complexity of ADM1, which must exhibit the same or an even richer dynamics than the model in \cite{WeeSeoWolk:2013}.

Bornh\"oft et al. \cite{Bornhoft:2013} introduced a model with five
state variables based on their observations from a   numerical steady-state analysis
of the ADM1 model, and conjectured that their model undergoes the same
bifurcations as the ADM1 model with the substrate inlet concentration as
bifurcation parameter.   The model in  \cite{Bornhoft:2013} is
the first simplified model to  consider the effects of ammonia. 
It is demonstrated that the proposed model 
is able to capture the  same effects of ammonia
on anaerobic digestion that are displayed by the ADM1 model.
However, the analysis in this paper shows that the model does not
possess all   of the  dynamics   of  ADM1, even if a broader class of growth
functions is considered than the ones that were initially proposed.
The model is missing some of
the dynamics shown in \cite{WeeSeoWolk:2013}, namely the possible
bistability between two equilibria that both correspond to biogas
production, a behaviour of the full ADM1 model that is also noted in \cite{Bornhoft:2013}.

The model in \cite{Bornhoft:2013} considers two stages of  anaerobic digestion, acidogenesis and methanogenesis. In the first stage, simple substrates are broken down by acidogenic microorganisms. The microorganisms use the energy from the simple substrates to grow, and produce volatile fatty acids (VFAs) and ammonia as byproducts. The VFAs and ammonia have opposing effects on the pH of the system;  an increase in the concentration of VFAs will decrease the pH, while an increase in the concentration of ammonia will increase the pH.  In the second stage, methanogenic microorganisms convert the VFAs to biogas. The methanogenic microorganisms are very sensitive to the environment, and can only tolerate a relatively narrow pH range. Furthermore, ammonia is  toxic to the methanogenic microorganisms in large quantities and will restrict their growth. The flow chart in \cref{fig:flowchart} summarizes the process.
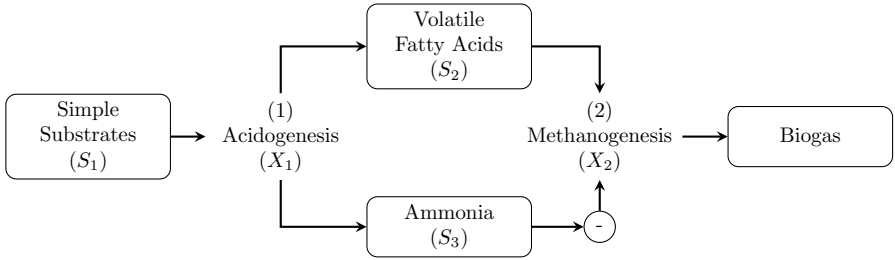
\begin{figure*}
\centering
\begin{tikzpicture} [node distance=2cm, scale=0.8, every node/.style={scale=0.8}]
\node (substrates) [box] {Simple Substrates\\ $(S_1)$};
\node (acid) [right of =substrates,xshift= 1.2cm,,text width=2.25cm, text centered] {(1)\\ Acidogenesis \\$(X_1)$};
\node (VFA) [box, right of =acid,xshift = 0.8cm,yshift=1.5cm] {Volatile Fatty Acids\\ $(S_2)$};
\node (ammonia) [box, right of =acid,xshift = 0.8cm, yshift=-1.5cm] {Ammonia\\ $(S_3)$};
\node (meth) [right of =VFA, xshift = 0.5cm, yshift=-1.5cm,text centered, text width =2.5cm] {(2)\\ Methanogenesis \\ $(X_2)$};
\node (biogas) [box, right of =meth,xshift = 1.5cm] {Biogas};
\node (inhibit) [ball, below of = meth,yshift = 0.5cm]{-};
\draw [arrow] (substrates) -- (acid);
\draw [arrow] (acid) |- (ammonia);
\draw [arrow] (acid) |- (VFA);
\draw [arrow] (VFA) -| (meth);
\draw [arrow] (ammonia) -- (inhibit);
\draw [arrow] (inhibit) -- (meth);
\draw [arrow] (meth) -- (biogas);
\end{tikzpicture}
\caption{\label{fig:flowchart} The anaerobic digestion process. (1) Acidogenic microorganisms break down simple substrates into VFAs and ammonia. (2) Methanogenic microorganisms break down VFAs into biogas such as methane. This process is inhibited by ammonia.}
\end{figure*}

In this paper we   provide  a formal mathematical analysis of the model
proposed in \cite{Bornhoft:2013},   allowing a more general class of response functions. In \cref{sec:model}, we
describe the model and assumptions, and give
 properties of the 
solutions  of the  system.
If the substrate input concentration is
 too low, the system converges to a
state where no microorganisms are present.
We show that if the  substrate input concentration is
high enough to allow the acidogenic microbial population
to survive, the system reduces to a
limiting system that is a two-dimensional basic model of growth in the chemostat that includes the decay rates and allows for any non-monotone response function.

  In \cref{sec:stability}, we study the dynamics
of this limiting system.
 We obtain a new global result in the case that the parameters
allow  bistability by proving that no nontrivial periodic orbits exist.

In \cref{sec:globalfull} we use the theory of asymptotically autonomous systems and
 the results for the limiting system from \cref{sec:stability}
to provide a complete global analysis of the anaerobic digestion model
in \cite{Bornhoft:2013} for a more general class of response functions.

In \cref{sec:bifurcations},  we propose two
alternative  prototype functions
 to model the growth of the methanogenic archaea and capture the inhibition by ammonia.
These prototypes  complement the
one used in \cite{Bornhoft:2013}, which has the property that
there is no growth in the absence of ammonia. The prototypes we
introduce allow growth in the absence of ammonia, but are either
unimodal or decreasing in ammonia.
We  provide bifurcation
diagrams for all three prototypes,  and compare how they
influence the outcome.

In \cref{sec:stochastics} we further investigate the model
  when the parameters are selected so that there are two stable steady states.
In industrial applications of processes
such as anaerobic digestion, operators must be aware of how physical and
environmental processes, and changes in the biology of the species can
affect the long term health of the reactor. One way to
address these
challenges is to include the   effects of stochasticity in
simulations of  the model.  Stochasticity
can be a result of random births and deaths, or of fluctuations in the
model parameters, possibly due to mutations  or changes in the environment.
Models of chemostats that
include stochasticity have been considered in the literature
(e.g., \cite{Campillo:2011,Imhof:2005}),
but our approach differs from the ones presented in those papers. We
consider stochasticity in the case where there
are fluctuations in the parameters, and compare the results to two well
known methods for simulating stochasticity in the case where there are
random births and deaths. Our studies give new insight into why
seemingly identical reactor setups can lead to very different reactor performances.
In \cref{sec:conclusion},  we summarize our results  and
discuss their implications for biogas production using anaerobic
digestion.

\section{The Model}\label{sec:model}
 Let $X_1$, $X_2$, $S_1$, $S_2$ and $S_3$ denote the concentrations of the acidogenic microorganisms, methanogenic microorganisms, simple substrates, acetic acid and ammonia, respectively. The model is described by the system
\begin{subequations}\label{eq:system}
	\begin{align}
		\dot{S_1} &= (S^{(0)}-S_1)D-y_1\mu_{1}(S_1)X_1, \label{eq:substrates}\\
		\dot{X_1} &= -D_1X_1+\mu_{1}(S_1)X_1,\label{eq:acidogens}\\
		\dot{S_2} &= -DS_2+y_2\mu_1(S_1)X_1-y_3\mu_2(S_2,S_3)X_2,	\label{eq:aceticacid}\\
		\dot{S_3} &= -DS_3+y_4\mu_1(S_1)X_1,\label{eq:ammonia}\\
		\dot{X_2} &= -D_2X_2 + \mu_2(S_2,S_3)X_2,\label{eq:methanogens}
	\end{align}
\end{subequations}
where  $D$ is the dilution rate, $S^{(0)}$ is the input concentration of
simple substrates, $D_i = D + k_i$, where $k_i\geq0$ are the respective
decay rates of $X_i$, and $y_i$ are yield constants.

Let $\mathbb{R}_+$ and $\mathbb{R}^2_+$  denote the set
of non-negative real numbers and  the non-negative plane,
respectively. We make the following assumptions concerning $\mu_1$ and $\mu_2$:

\begin{enumerate}[label=(H\arabic*)]
	\item $\mu_1(S_1)\in C^1(\mathbb{R}_+)$, and $\mu_1'(S_1)>0$ for all $S_1>0$.
	\item $\mu_1(0) = 0$, $\mu_1(S_1)>0$ for all $S_1>0$.
	\item $\mu_2(S_2,S_3)\in C^1(\mathbb{R}^2_+)$, and $\mu_2(S_2,S_3)>0$ if $S_2>0$ and $S_3>0$.	
	\item $\lim_{S_3\to\infty}\mu_2(S_2,S_3) = 0$ for all $S_2\geq 0 $.
	\item $\lim_{S_2\to\infty}\mu_2(S_2,S_3) = 0$ for all $S_3\geq 0 $.
	\item $\mu_2(0,S_3) = 0$ for all $S_3\geq0$ and $\mu_2(S_2,0)\geq 0 $ for all $S_2>0$
	\item There exists $\Gamma(S_3)\in C(\mathbb{R}_+)$ such that for $S_2<\Gamma(S_3),~\partial_{S_2}\mu_2(S_2,S_3)>0$ and for $S_2>\Gamma(S_3),$ $\partial_{S_2}\mu_2(S_2,S_3)<0$.
	\end{enumerate}
Unlike in \cite{Bornhoft:2013}, we do not assume that both $X_i$ have
identical decay rates.
(H1) and (H2) are satisfied by any of the Holling type I, II or III
growth functions, which are standard to chemostat models. (H3), (H4), and
(H5) capture the inhibitory nature of $S_2$ and $S_3$, guaranteeing that
large quantities of either $S_2$ or $S_3$ will be detrimental to the
growth of $X_2$. (H6) ensures that an absence of acetic acid will result
in no growth of the methanogenic microorganisms, while an absence of
ammonia does not necessarily have this effect. We would like to note
that the prototype describing the growth of methanogens proposed in
\cite{Bornhoft:2013} has the property that \mbox{$\lim_{S_3 \to 0}
\mu_2(S_2,S_3)=0$} for all $S_2\geq 0$. We decided to relax this
condition. (H7) intends to capture the nature of the inhibition
mechanisms outlined in \cite{Bornhoft:2013}, whereby small
concentrations of $S_2$ are limiting on the growth of $X_2$, while large
concentrations of $S_2$ will increase the pH and hence be inhibitory.
 The curve $\Gamma(S_3)$ indicates that for each fixed $S_3$ there is at most one $S_2$ such that $\partial_{S_2}\mu_2(S_2,S_3)=0$. The curve $\Gamma(S_3)$ therefore gives the optimal concentration of $S_2$ for growth of $X_2$ as a function of $S_3$. We make no further assumptions about how
$\mu_2(S_2,S_3)$ changes with $S_3$. In many cases, including ADM1, $\mu_2(S_2,S_3)$ will have a unimodal shape for fixed $S_2$, but we
  do not want to rule out the possibility of other profiles that may be
  useful.  For examples of functions satisfying these
  hypotheses, refer to \cref{sec:bifurcations}.

Here, we introduce some notation. The \emph{break-even concentration of
$S_1$}, $\lambda_1$, is the unique positive extended real number that
satisfies
\begin{equation}
 \mu_1(\lambda_1) = D_1.
\end{equation}
If no such number exists, we take $\lambda_1 = +\infty$.  When $S_1 = \lambda_1$ and $X_2 = 0$, the equilibrium concentrations of $S_2$ and $S_3$, $\lambda_2$ and $\lambda_3$, respectively, are given by
\begin{align}
	\lambda_2 &= \frac{y_2}{y_1}(S^{(0)}-\lambda_1)
	\label{eq:lambda2}\\
	\lambda_3 &= \frac{y_4}{y_1}(S^{(0)}-\lambda_1)
	\label{eq:lambda3}.
\end{align}

The \emph{break-even concentrations of $S_2$}, $\sigma_1$ and
$\sigma_2$, are the extended real numbers $\sigma_2\geq\sigma_1$ that solve
\begin{equation}
 \mu_2(\sigma_i,\lambda_3) = D_2.
\end{equation}
If no such numbers exist, which is the case when
$\mu_2(\Gamma(S_3),S_3)<D_2$, then we write $\sigma_1=\sigma_2=+\infty$.
 Note that by (H7), there is at most one turning point of
$\mu_2(S_2,\lambda_3)$, and therefore at most two solutions to $\mu_2(S_2,\lambda_3) = D_2$.

System \cref{eq:system} has a total of four possible equilibria,
\begin{subequations}
\begin{align}
	E &= (S^{(0)},0,0,0,0) \label{eq:E}\\
	E_0 &= (\lambda_1,X_1^*,\lambda_2,\lambda_3,0) \label{eq:E0}\\
	E_1&= (\lambda_1,X_1^*,\sigma_1,\lambda_3,X_{2,\sigma_1}^*)
	\label{eq:E1}\\
	E_2&= (\lambda_1,X_1^*,\sigma_2,\lambda_3,X_{2,\sigma_2}^*),
	\label{eq:E2}
\end{align}
\end{subequations}
where
$$   X_{1}^*
=\frac{D(S^{(0)}-\lambda_1)}{y_1D_1}   \quad \mbox{and} \quad X_{2,\sigma_i}^*
=\frac{D(\lambda_2-\sigma_i)}{y_3D_2}. $$
These equilibria are only biologically meaningful if each of the components is non-negative. $E_1$ and $E_2$, when they exist, are called interior equilibria, since they lie in the interior of the positive cone $\mathbb{R}_+^5$. $E$ and $E_0$ are called boundary equilibria, since they lie on the boundary of the positive cone $\mathbb{R}_+^5$.

  The following propositions give well-posedness results for
system \cref{eq:system}, provide conditions for the washout of the
microorganisms in the reactor when the  substrate input concentration is too low,
and introduce the limiting system when the input concentration is high
enough so that  the acidogens survive.  The proofs are given in \cref{ap:quasiproof}.

\begin{prop}\label{prop:properties} Assume that $X_i(0)\geq 0$ and $S_i(0)\geq 0 $.
\begin{enumerate}[label=\roman*)]
	\item If $X_1(0) = 0$, then solutions converge to $E$ as $t \to \infty$.
	\item If $X_1(0)>0$ and $X_2(0)=0$, then $X_1(t)>0$ and $S_i(t)>0$ for all $t>0$  while $X_2(t) = 0$ for all $t\geq 0$.
 	\item If $X_1(0)>0$ and $X_2(0)>0$, then $X_i(t)>0$ and $S_i(t)>0$ are positive for all $t>0$.
	\item All solutions are bounded for $t\geq 0$.
\end{enumerate}
\end{prop}

\begin{prop} If $\lambda_1 \geq S^{(0)}$, then $E$ is a globally asymptotically stable equilibrium of \cref{eq:system}. \label{prop:washout}
\end{prop}

\begin{prop}\label{prop:limit}
	If $\lambda_1 < S^{(0)}$, then \cref{eq:system} is a quasi-autonomous system with limiting system
	\begin{subequations}\label{eq:limitsystem}
	\begin{align}
	\dot{S_2} &= -DS_2+\lambda_2D-y_3\mu_2\left(S_2,\lambda_3\right)X_2,\label{eq:limitacid}\\
	\dot{X_2} & = -D_2X_2+\mu_2(S_2,\lambda_3)X_2.\label{eq:limitmethanogens}
	\end{align}
	\end{subequations}
\end{prop}
By Theorem~1.4 in \cite{Th:1994}, it will
be enough to study the dynamics of this limiting system.

\section{ Global Analysis of Growth in  the Chemostat}\label{sec:stability} 
After the change of variables
\begin{align*}
 X(t) =  y_3 X_2(t), \qquad S(t) =  S_2(t),\qquad  \mu(S(t)) = \mu_2(S_2(t),\lambda_3),\qquad S^{0} = \lambda_2,
\end{align*}
system \cref{eq:limitsystem} becomes a model of the chemostat:
\begin{subequations}\label{eq:Chemostat}
 \begin{align}
  \dot{S}(t) &= -(S(t)-S^{0})D - \mu(S(t))X(t)\\
  \dot{X}(t) &= -D_2X(t)+\mu(S(t))X(t).
 \end{align}
\end{subequations}

Recall that $\mu_2(S(t),\lambda_3)$ satisfies (H3) and (H4), and hence
 $\mu(S(t))$ is a non-monotone response function with break-even
concentrations    $0<\sigma_1<\sigma_2$,   the   extended real numbers that
solve $\mu(\sigma_i) =  D_2$.

  We  define the equilibria of \cref{eq:Chemostat} that
correspond to  $E_0, E_1,$ and $E_2$, respectively, for system
\cref{eq:system} defined in
\cref{eq:E0}-\cref{eq:E2}:
$$ \mathcal{E}_0 = (S^0, 0), \quad
\mathcal{E}_1 = (\sigma_1,X_{\sigma_1}^*), \quad
\mathcal{E}_2 = (\sigma_2,X_{\sigma_2}^*).$$
where $X^*_{\sigma_i} = \displaystyle \frac{D(S^0-\sigma_i)}{D_2}, \ i=1,2$.

 Models of the chemostat have been well studied (e.g., see
\cite{SW:1995,Harmand:2017} and
the references therein).
Model \cref{eq:Chemostat}  is a model of growth of a single species in
the   chemostat with non-monotone response
function that  includes the  species decay rate, i.e. $D_2=D+\epsilon$
where $\epsilon>0$ is the species decay rate.

In  Wolkowicz and Lu \cite{WL:1992}, model \cref{eq:Chemostat}
extended to the $n$ species case was analyzed.
 The results of that paper, if applied to the single species growth model,
completely determine the dynamics of \cref{eq:Chemostat} when $\mu(S(t))$ is any monotone increasing function or it is non-monotone and $\sigma_1<S^0\leq \sigma_2$.
However, the case that $\mu(S(t))$ is  a  non-monotone response function
and $\sigma_1<\sigma_2<S^0$, remained open.
Here, we provide a proof in this case
and thus complete the global
analysis of system \cref{eq:Chemostat}.
  In particular, we prove that
  there are no periodic orbits, and hence although the outcome is
  initial condition dependent, either the species dies
  out or it approaches an equilibrium.

  In the following theorem, we summarize what is known for  the dynamics of
  \cref{eq:Chemostat}, and provide a proof in the case that had
  remained open.

\begin{thm} \label{th:chemostat}
  Consider model \cref{eq:Chemostat}. Assume $\mu(S)$ is continuously differentiable, $\mu(0)=0$,
	$\mu(S)\geq 0$ for all $S>0$, and there exist positive  
	numbers     $\sigma_1 \leq \sigma_2$ (possibly infinite) such that
	$\mu(S)<D_2$ if $0<S< \sigma_1$,   $\mu(S)>D_2$ if
	$ \sigma_1<S<\sigma_2$, and  	$\mu(S)<D_2$ if $S>\sigma_2$.
	Let $S(0)\geq 0$ and	$X(0)>0$. 
\begin{enumerate}[label = \roman*)]
 \item If $S^0\leq \sigma_1\leq \sigma_2$, then $\mathcal{E}_0$ is globally asymptotically stable.
 \item If $\sigma_1<S^0\leq \sigma_2$, then $\mathcal{E}_1$ is globally asymptotically stable.
	 	\item   If $\sigma_1 = \sigma_2 < S^0$, then $\mathcal{E}_0$ is locally
		asymptotically stable and attracts all solutions except
		the solutions in the stable manifold of
		$\mathcal{E}_1=\mathcal{E}_2$.
 \item If $\sigma_1 < \sigma_2 < S^0$, then $\mathcal{E}_1$ and
	 $\mathcal{E}_0$ are locally asymptotically stable  and
		$\mathcal{E}_2$ is a saddle. Furthermore, any orbit that
		is not in the stable manifold of $\mathcal{E}_2$
		converges to either $\mathcal{E}_1$ or $\mathcal{E}_0$.
\end{enumerate}
\end{thm}

\begin{proof}
	\textit{i)-ii)}  See
	\cite{WL:1992}.

	\textit{iii)}    This result follows from standard phase plane analysis.
When $\sigma_1=\sigma_2$, $\mathcal{E}_1$ and $\mathcal{E}_2$ coalesce and are unstable.
All orbits converge to $\mathcal{E}_0$ except those in the stable manifold of the
degenerate saddle  $\mathcal{E}_1=\mathcal{E}_2$.

\textit{iv)}
	 	If $\sigma_1< \sigma_2<S^0$, then
	$\mathcal{E}_0$, $\mathcal{E}_1$, and $\mathcal{E}_2$ all lie in
	$\mathbb{R}_+^2$.
 	From standard local stability analysis, it
	follows that $\mathcal{E}_0$ and $\mathcal{E}_1$ are both
	locally asymptotically stable and $\mathcal{E}_2$ is a saddle.

	Next we show that no nontrivial periodic solutions are possible.   We proceed
	using proof by contradiction. Suppose that
	there exists a nontrivial periodic solution, $\Phi$.
	By \cref{prop:properties} all solutions are bounded in forward
	time and the first
	quadrant is invariant.
	Orbits above the $\dot{S}=0$ nullcline (shown in \cref{fig:Lyapunov} by the dashed curve) move from right to left, and so if they cross the $\dot{X}=0$ nullclines (Shown in \cref{fig:Lyapunov} by dashed vertical lines) cross them from right to left. Orbits below the $\dot{S}=0$ nullcline move from left to right, and so if they cross the $\dot{X}=0$ nullclines, cross them from left to right. 
	 Orbits thatmare to the right  of the line $S=\sigma_2$ that meet the $S'=0$ nullcline cross it downward. Therefore, if  a periodic orbit exists,  it must lie entirely to the left of  $S=\sigma_2$, since if an orbit enters the region below the $S'=0$ nullcline and to the right of $S=\sigma_2$, it is  trapped in that  region and must converge to $\mathcal{E}_0$ by the Poincar\'e-Bendixson Theorem. If it is to the left of $S=\sigma_2$ and above $\dot{S}=0$ it moves to the left.
	Therefore, by the Poincar\'e-Bendixson Theorem and
	standard phase-plane analysis, $\Phi$ must surround
	$\mathcal{E}_1$, and must lie entirely in the set
	$$  \mathcal{G} = \{(S,X):
	\ 0<S<\sigma_2, \ X>0 \}.$$
	Define the Lyapunov function,
\begin{align}\label{eq:lyapunov}
V(S,X) = \int_{\sigma_1}^{S} \frac{(\mu(\xi)-D_2)(S^0-\sigma_1)}{D_2(S^0-\xi)}d\xi+\left[X-X_{\sigma_1}^*-X_{\sigma_1}^*\ln\left(\frac{X}{X_{\sigma_1}^*}\right)\right],
\end{align}
as in \cite{WL:1992}. See \cref{fig:Lyapunov} for
	 	phase portraits of system \cref{eq:Chemostat}
	with typical level sets of the Lyapunov function. Note that \cref{eq:lyapunov} is a valid Lyapunov function for $\mathcal{E}_1$ in $\mathcal{G}$, and
\begin{align}\label{eq:tderivative}
	\dot{V}(S,X) &=
	X(\mu(S)-D_2)\left(1-\frac{\mu(S)(S^0-\sigma_1)}{D_2(S^0-S)}\right),\notag
\end{align}
is non-positive, for all $S\in[0,\sigma_2]$ and $X\geq0$,
i.e.,  for all $S$ in the
closure of $\mathcal{G}$. By examining
	\begin{align*}
		\nabla V(S,X) &= \left(
		\frac{(\mu(S)-D_2)(S^0-\sigma_1)}{D_2(S^0-S)},1-\frac{X_{\sigma_1}^*}{X}\right)={\bf 0},
	\end{align*}
we see that $\mathcal{E}_1$ and $(\sigma_2,X_{\sigma_1}^*)$ are the only
critical points of $V(S,X)$ for which $S\leq \sigma_2$. The point $(\sigma_2,X_{\sigma_1}^*)$ is directly above $\mathcal{E}_2$ in phase
space, since by definition $X_{\sigma_1}^*> X_{\sigma_2}^*$. Notice that $\partial^2 V(S,X)/\partial X^2 = X_{\sigma_1}^*/X^2 > 0$ for all $X>0$, $\partial V(S,X)/\partial S < 0$ for $0<S<\sigma_1$, $\partial V(S,X)/\partial S >0$ for $\sigma_1<S <\sigma_2$, and $\partial V(S,X)/\partial S <0$ for $\sigma_2 < S<S^0$. It follows that 
$\mathcal{E}_1$ is a local minimum of $V(S,X)$, and $(\sigma_2,X_{\sigma_1}^*)$ is a saddle point of $V(S,X)$.
The level set  $V(S,X)=
	V(\sigma_2,X_{\sigma_1}^*) $ is given by
\begin{align*}
	V(\sigma_2,X_{\sigma_1}^*) =
	\int_{\sigma_1}^{\sigma_2}\frac{(\mu(\xi)-D_2)(S^0-\sigma_1)}{D_2(S^0-\xi)}
	d\xi.
\end{align*}
For $S\leq \sigma_2$,   this level set   is a closed curve surrounding $\mathcal{E}_1$ that passes through
the point  $(\sigma_2,X_{\sigma_1}^*)$. (See the bold level set in
\cref{fig:Lyapunov} where  two possible
configurations are shown.) Since $\dot{V}(S(t),X(t))\leq 0$, the set
\begin{align*}
\mathcal{U} = \{(S,X)\in \mathbb{R}_+^2: 0\leq S\leq\sigma_2, V(S,X) \leq V(\sigma_2,X_{\sigma_1}^*)\}
\end{align*}
is a positively invariant set. Since any periodic orbit must lie entirely in $\mathcal{G}$ and must
surround an equilibrium point, it must enter  $\mathcal{U}$.
Since $\mathcal{U}$  positively invariant, it follows that $\Phi$ is
contained entirely in $\mathcal{U}$. By the   minor
variation of  LaSalle's
invariance principle \cite{WL:1992}, any trajectory in $\mathcal{U}$ converges to the largest invariant set in $\mathcal{U} \cap \{(S,X): \dot{V}(S,X) =0\}$.
The only such invariant set is $\mathcal{E}_1$, and hence no nontrivial
periodic orbit exists, a contradiction.

Now, from standard phase plane analysis and the Poincar\'{e}-Bendixson
Theorem, it follows that all orbits converge to one of the three equilibria as $t$ tends to
infinity. The one-dimensional stable manifold of $\mathcal{E}_2$ acts as a separatrix, defining the basins of attraction for $\mathcal{E}_1$ and $\mathcal{E}_0$.
\qed
\end{proof}

\begin{figure}
\begin{centering}
\subfigure[]{\includegraphics[width=0.4\textwidth]{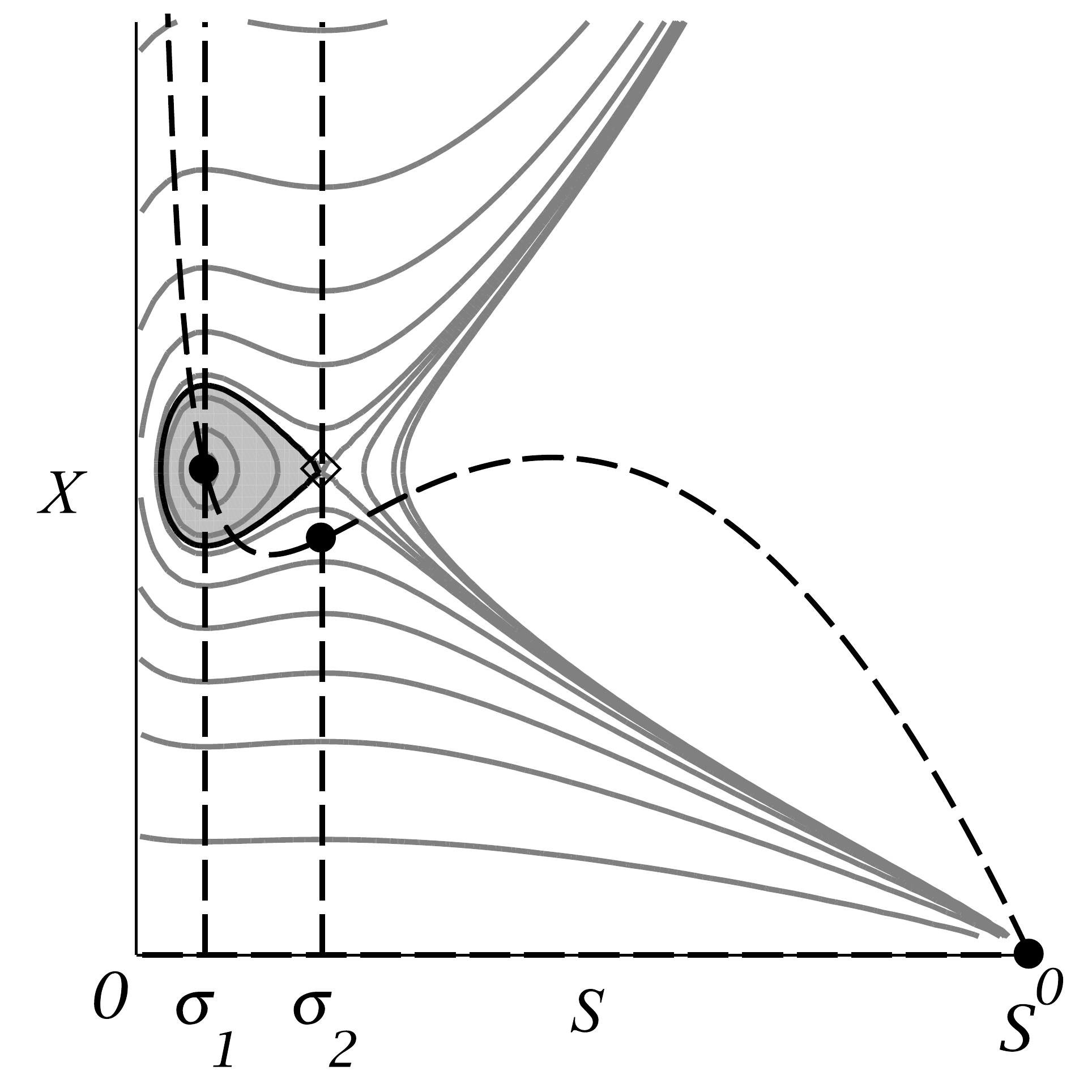}}
\subfigure[]{\includegraphics[width=0.4\textwidth]{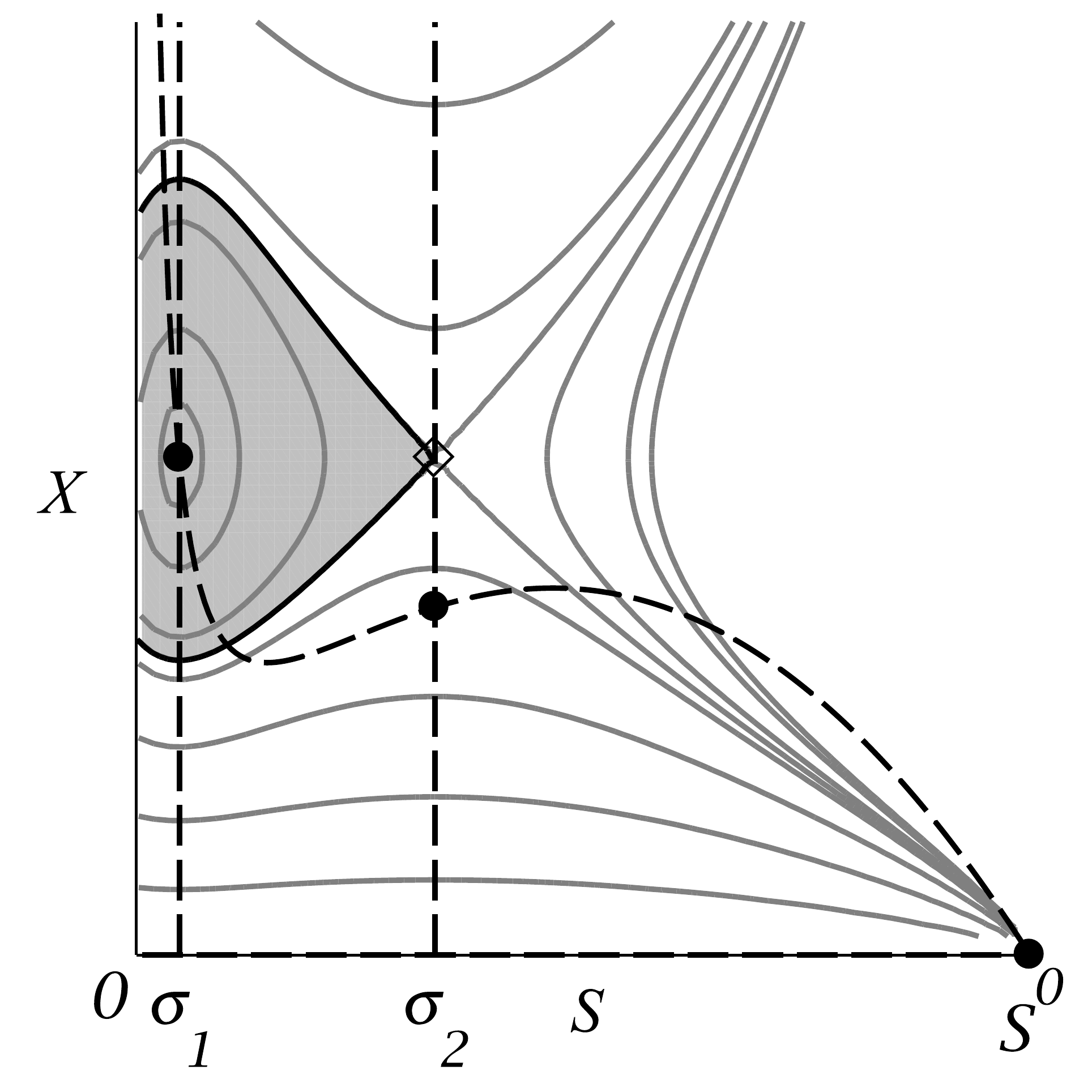}}
\caption{Phase portraits of system 
\cref{eq:Chemostat}
with the level sets of $V(S,X)$. The dashed lines are the nullclines for $X$ and the dashed curve is the nullcline for $S$. The equilibria
	$\mathcal{E}_0$, $\mathcal{E}_1$, and $\mathcal{E}_2$ are
	indicated by filled circles, and the point $(\sigma_2,X_{\sigma_1}^*)$ is
	indicated by an open diamond. The grey curves are the level sets of
	$V(S,X)$. The bold curve is the level set of $V(S,X)$ that passes through the point $(\sigma_2,X_{\sigma_1}^*)$. The set $\mathcal{U}$ is shaded in gray. These figures were produced using Maple \cite{maple}. \label{fig:Lyapunov}}
\end{centering}
\end{figure}

\section{Global Analysis of the Full System
\cref{eq:system} }\label{sec:globalfull}

\begin{thm}\label{thm:fullsystem}
	Consider model \cref{eq:system}.
\begin{enumerate}[label = \roman*)]
	\item If $\lambda_1 \geq S_1^{(0)}$, then $E$ is globally
		asymptotically stable.
	\item If $\lambda_1<S_1^{(0)}$ and $\lambda_2 \leq
		\sigma_1\leq\sigma_2$, then $E_0$ is a globally asymptotically stable equilibrium.
	\item If $\lambda_1<S_1^{(0)}$ and $\sigma_1<\lambda_2 \leq
		\sigma_2$, then $E_1$ is a globally asymptotically stable equilibrium.
	\item If $\lambda_1<S_1^{(0)}$ and $\sigma_1<
		\sigma_2<\lambda_2$, then $E_0$ and $E_1$ are locally
		asymptotically stable, and $E_2$ is a saddle.
		Furthermore any orbit that does not lie on the stable manifold of $E_2$ converges to one of $E_0$ or $E_1$.
\end{enumerate}
\end{thm}
\begin{proof}
	\textit{i)} was proved in \cref{prop:washout}.
	
	\textit{ii) - iv)} Since each of the $E_i, ~i=0,1,2$ for model \cref{eq:system} corresponds to
$\mathcal{E}_i$ for system \cref{eq:Chemostat},  the
	results follow   from the results for the limiting
	system given in \cref{th:chemostat}, followed by an application of the
	theory for asymptotically autonomous systems, either by using
	Theorem~1.4
in \cite{Th:1994}, or by a direct proof using the Butler-McGehee Lemma (as
	stated in Lemma~5.2 in  \cite{BW:1985} and applied there).
\qed
\end{proof}

	\section{Bifurcation Analysis of Full System~\cref{eq:system}}\label{sec:bifurcations}

As a result of the analysis of the previous two sections, the only possible bifurcations that can occur in \cref{eq:system} are transcritical bifurcations and saddle-node bifurcations.

	\begin{figure*}
	\begin{centering}
	\subfigure[]{\includegraphics[width=.4\textwidth]{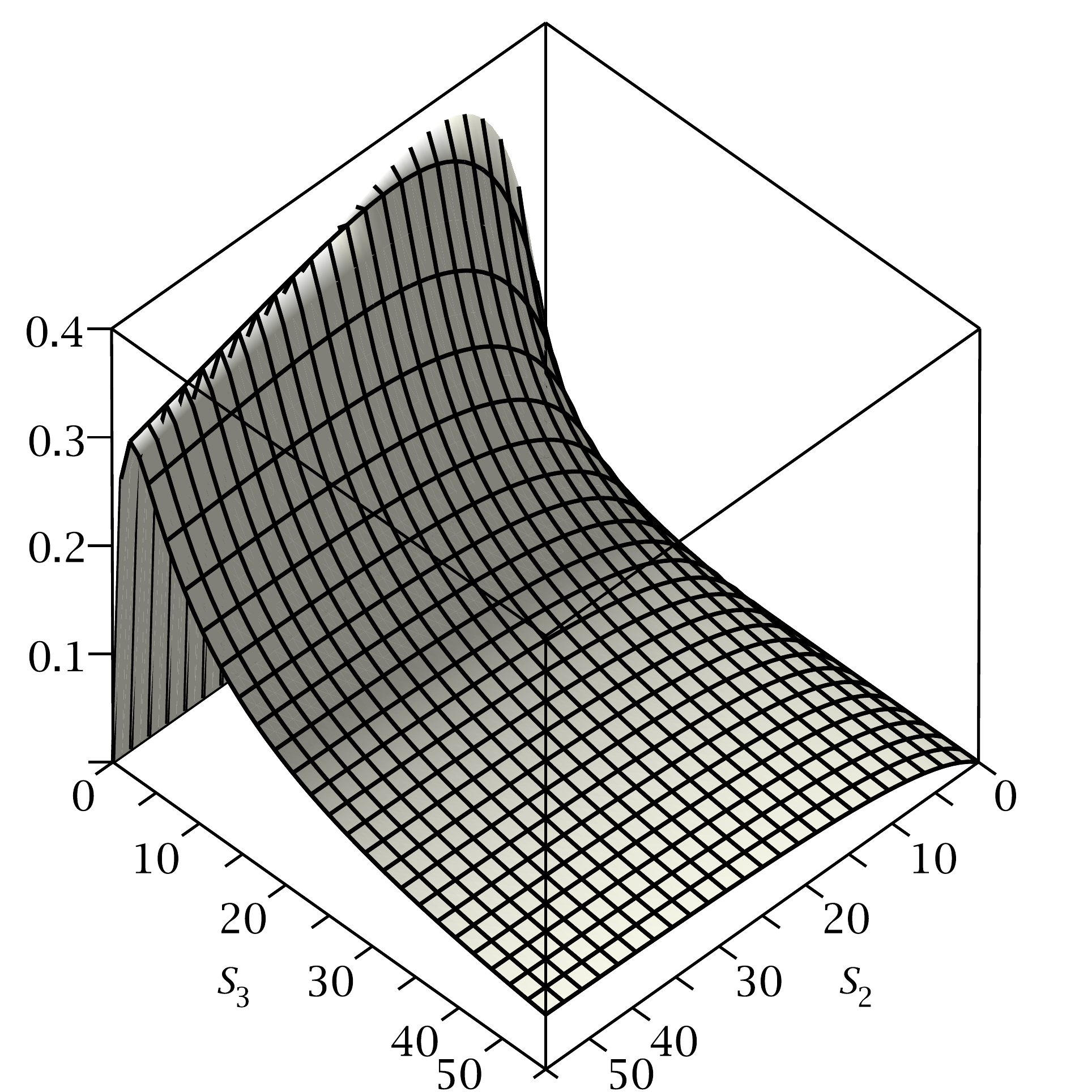}}
	\subfigure[]{\includegraphics[width=.4\textwidth]{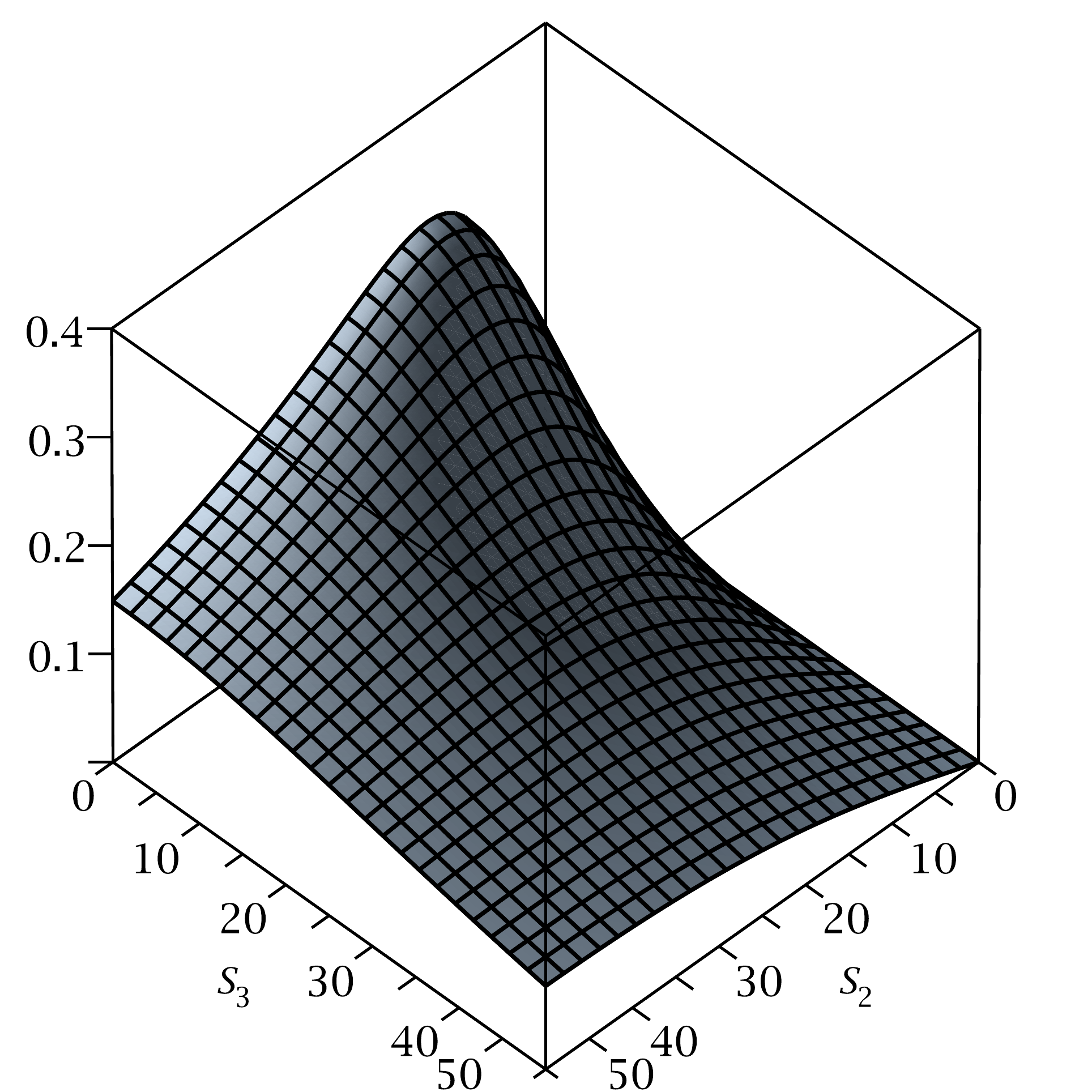}}
	\subfigure[]{\includegraphics[width=.4\textwidth]{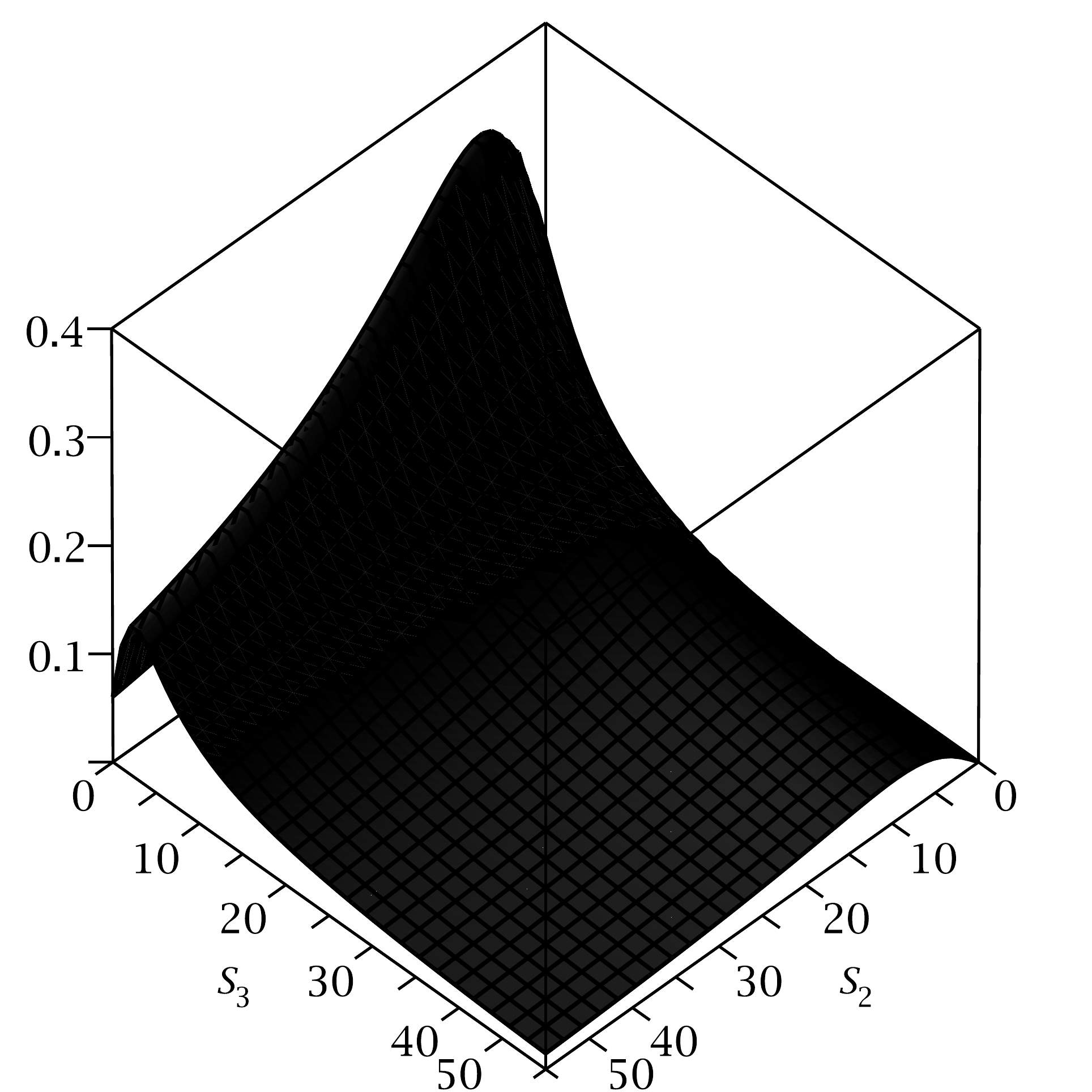}}
	\subfigure[]{\includegraphics[width=.4\textwidth]{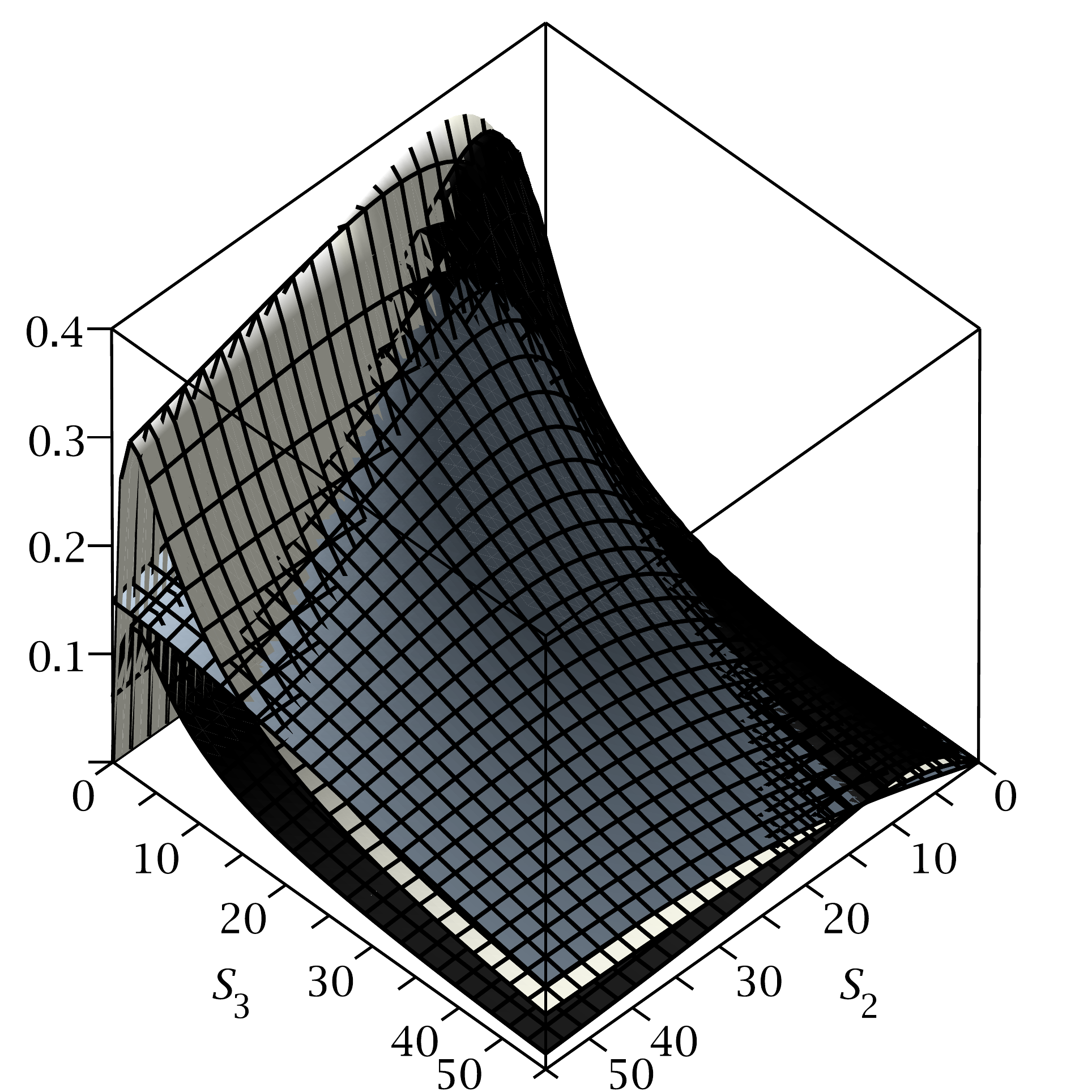}}
\caption{Plots of each prototype in 3-dimensions: (a)
$\proto{I}(S_2,S_3)$, (b) $\proto{II}(S_2,S_3)$, (c)
$\proto{III}(S_2,S_3)$. (d) shows all three prototypes on the same axes
for comparison.  Parameters values used are given in
\cref{tab:parameters}. \label{fig:prototypes}}\end{centering}
\end{figure*}

In \cite{Bornhoft:2013}, a prototype growth function was introduced to capture the inhibition caused by ammonia. This prototype
\begin{equation}
\proto{I}(S_2,S_3) = \frac{m_{I}S_2S_3}{(K+S_2)(S_3+k_1S_2)(1+k_2S_3)},
\end{equation}
has the property that when there is no ammonia, which is toxic to the methanogenic microorganisms, the methanogenic microorganisms are unable to grow. We introduce two additional prototype functions
\begin{subequations}
\begin{align}
 \proto{II}(S_2,S_3) &= \frac{m_{II} S_2}{K+k_1(S_2-S_3)^2+rS_2S_3},\\
 \proto{III}(S_2,S_3)&= \frac{m_{III} S_2(1+S_3)}{(K+k_1S_2+rS_2^2)(a+S_3^2)},
\end{align}
\end{subequations}
that satisfy (H3)-(H7). Both $\proto{II}(S_2,S_3)$ and
$\proto{III}(S_2,S_3)$ satisfy the additional property,
$\proto{}(S_2,0) \geq 0$ with equality only when $S_2=0$ or in the limit
as $S_2 \to \infty$.   For the parameters given in
\cref{tab:parameters}, $\proto{II}(S_2,S_3)$ is strictly decreasing
in $S_3$ and can be thought of as the opposite extreme of
$\proto{I}(S_2,S_3)$. It describes the scenario where ammonia is
strictly inhibitory and the methanogenic microorganisms do best without
any ammonia present. With a different set of parameters this response
function can be unimodal in $S_3$. The difference term in
the denominator of $\proto{II}(S_2,S_3)$ acts as a proxy for the
influence of pH in
the system. Since $S_2$ is acidic, and $S_3$ is basic, we assume that a
large difference between the two concentrations would cause the pH to be outside of the acceptable range for the growth of $X_2$.
The third prototype, $\proto{III}(S_2,S_3)$ covers the middle ground
between $\proto{I}(S_2,S_3)$ and $\proto{II}(S_2,S_3)$; it is unimodal
in $S_3$, like $\proto{I}(S_2,S_3)$, but is non-zero when $S_2>0$ and
$S_3=0$, like $\proto{II}(S_2,S_3)$.

The  substrate input concentration, $S^{(0)}$, and
dilution rate, $D$, are the two parameters that the operator of
a reactor has the ability to control.
In our  bifurcation analysis, we focus on how the dynamics of the full system
\cref{eq:system}  change when these parameters vary.
We note that  $\lambda_2$ and $\lambda_3$ depend on $S^{(0)}$ (see
\cref{eq:lambda2,eq:lambda3}),
and hence,
  $\max_{S_2>0}\mu_2(S_2,\lambda_3)$ changes when $S^{(0)}$ changes. From the
stability analysis in \cref{sec:stability}, two scenarios are
possible. In the first scenario (see \cref{fig:bifurcation}), there
is a transcritical bifurcation when $\lambda_2 = \sigma_1$, a transcritical
bifurcation when $\lambda_2 = \sigma_2$, and a saddle-node bifurcation
when $\max_{S_2>0}\mu_2(S_2,\lambda_3) = D_2$. In the second scenario (see
\cref{fig:bifloop}) there are two saddle-node bifurcations as
$\lambda_3$ increases. This sequence of bifurcations occurs because
 $\proto{I}(S_2,S_3)$ and $\proto{III}(S_2,S_3)$
are unimodal in $S_3$. With the parameters listed in \cref{tab:parameters}, the
second prototype, $\proto{II}(S_2,S_3)$, is strictly decreasing in
$S_3$, and so only the first scenario is possible. The other two
prototypes, $\proto{I}(S_2,S_3)$ and $\proto{III}(S_2,S_3)$, are unimodal
in $S_3$, and  either scenario is possible.

In the  bifurcation diagrams
  shown in \cref{fig:bifurcation,fig:bifloop},
\begin{equation} \mu_1(S_1) = \frac{\kappa S_1}{r_1+S_1},
\end{equation} and the parameters are the ones used in \cite{Bornhoft:2013}. Any  parameters
not given
in   \cite{Bornhoft:2013} (e.g.,
$m_{II}$, $m_{III}$, $r$, and $a$),  were chosen so that the functions, $\proto{II}$ and
$\proto{III}$, closely resemble the function $\mu_{2,I}$ given in
\cite{Bornhoft:2013}.   See \cref{tab:parameters} for the parameter
values used.   A plot of each function is shown in \cref{fig:prototypes}.
The bifurcation diagrams in \cref{fig:bifurcation} are qualitatively
similar for each uptake function. The bifurcation diagrams corresponding
to $\proto{II}(S_2,S_3)$ and $\proto{III}(S_2,S_3)$  resemble the
diagram for ADM1 in \cite{Bornhoft:2013} more closely than the diagram for $\proto{I}(S_2,S_3)$.

\def\arraystretch{1.35}%
\begin{table}
{\footnotesize
\begin{center}
\caption{The parameter values  used in the
	following bifurcation diagrams are the ones used in \cite{Bornhoft:2013},
	except  $m_{II},
	m_{III}, r,$ and $a$, which were chosen so that the response
	functions $\mu_{2,II}$ and $\mu_{2,III}$ closely resemble
	$\mu_{2,I}$.  The parameter
     $D$ is the bifurcation parameter in
	\cref{fig:bifa,fig:bifc,fig:bife}, and $S^{(0)}$ is the bifurcation
	parameter in \cref{fig:bifb,fig:bifd,fig:biff}.
	\label{tab:parameters} }
\begin{tabular}{|c|c|c|c|c|c|c|c|c|c|} \hline
Parameter & $S^{(0)}$& $D$ & $D_i, i=1,2$ & $\kappa$ & $K$ & $k_1$ & $k_2$ & $r$ & $r_1$ \\ \hline
Value & 50 & 0.15&0.16 &1.2 & 9.28 &0.05 &0.5 &0.1 &7.1\\\hline
\end{tabular}
\begin{tabular}{|c|c|c|c|c|c|c|c|c|} \hline
Parameter & $m_I$ & $m_{II}$& $m_{III}$&$y_1$&$y_2$&$y_3$&$y_4$& $a$     \\\hline
Value &1.64&0.4&3& 42.14& 116.5& 268 & 1.165& 12  \\ \hline
\end{tabular}
\end{center}}
\end{table}

In the diagrams where $D$ was used as the bifurcation parameter
(\cref{fig:bifa,fig:bifc,fig:bife}), there are three clear regions. In the first
region when $0<D<D^*_1$, only the equilibria $\mathcal{E}_1$ and
$\mathcal{E}_0$ lie in the positive cone, $\mathcal{E}_1$ is globally
asymptotically stable and therefore all non-stationary solutions
converge to $\mathcal{E}_1$. When $D=D^*_1$ the washout equilibrium
$\mathcal{E}_0$ undergoes a transcritical bifurcation. In the second
region, where $D^*_1<D<D_2^*$ all three equilibria lie in the positive
cone. $\mathcal{E}_1$ and $\mathcal{E}_0$ are locally asymptotically
stable and $\mathcal{E}_2$ is a saddle. All solutions (except the stable
manifold of $\mathcal{E}_2$) converge to one of $\mathcal{E}_1$ or $\mathcal{E}_0$, depending on initial conditions. When $D=D_2^*$, the two interior equilibria $\mathcal{E}_1$ and $\mathcal{E}_2$ undergo a saddle-node bifurcation. In the third region, where $D_2^* < D$ only $\mathcal{E}_0$ exists, and it is globally asymptotically stable. Therefore all solutions tend to $\mathcal{E}_0$.

\begin{figure*}
\centering
	\subfigure[\label{fig:bifa}]{\includegraphics[width=.3\textwidth]{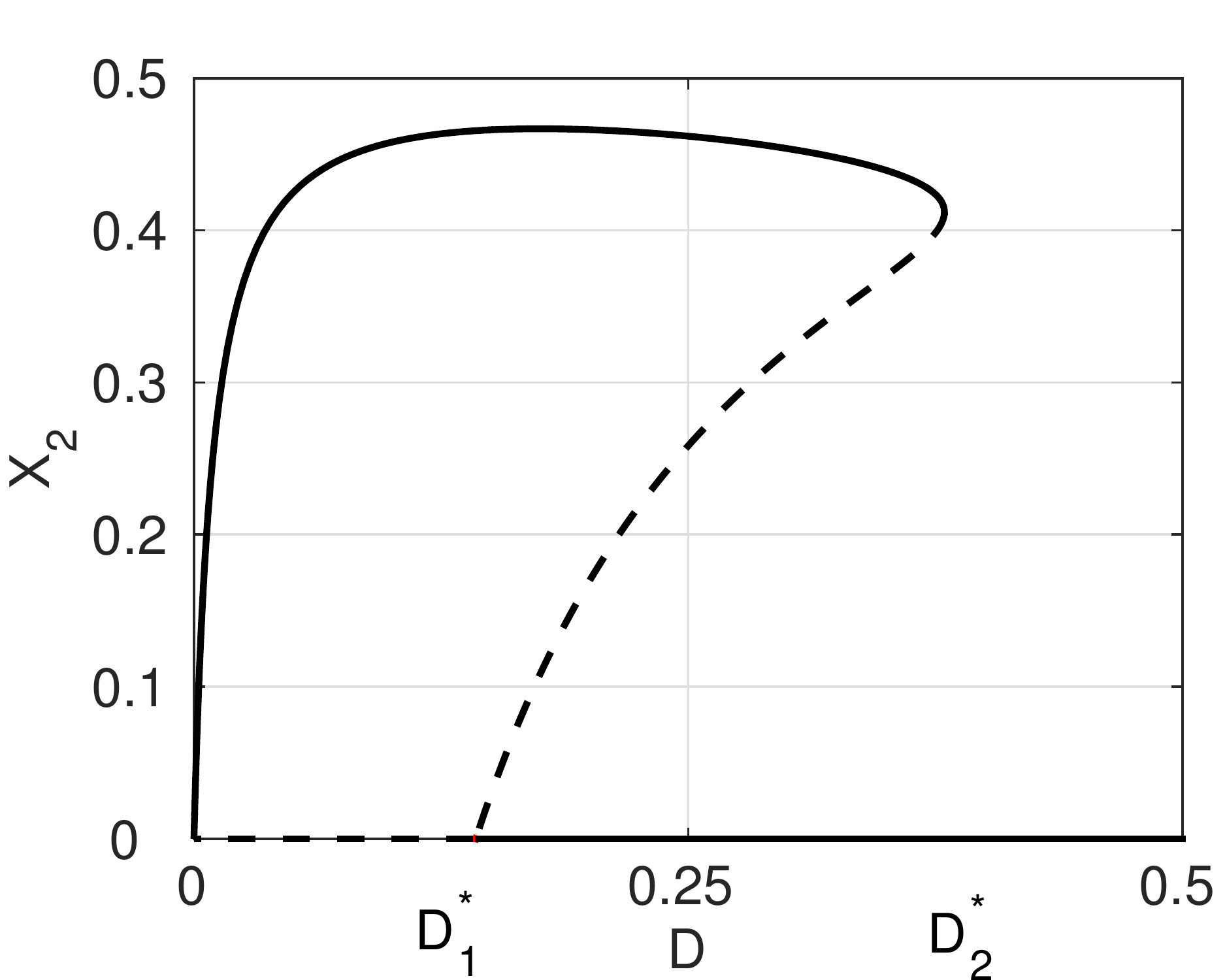}}
	\subfigure[\label{fig:bifb}]{\includegraphics[width=.3\textwidth]{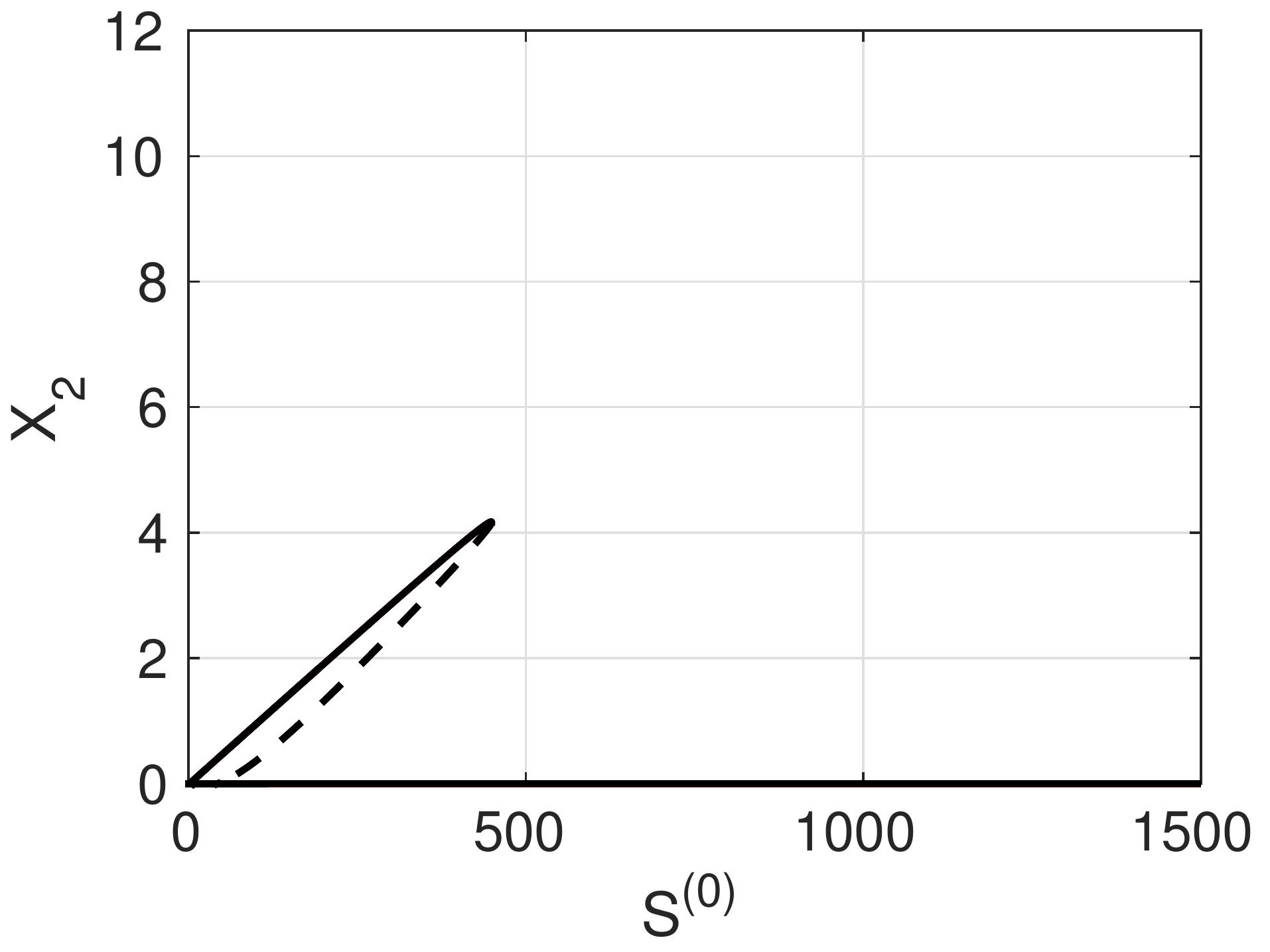}}
	\subfigure[\label{fig:bifc}]{\includegraphics[width=.3\textwidth]{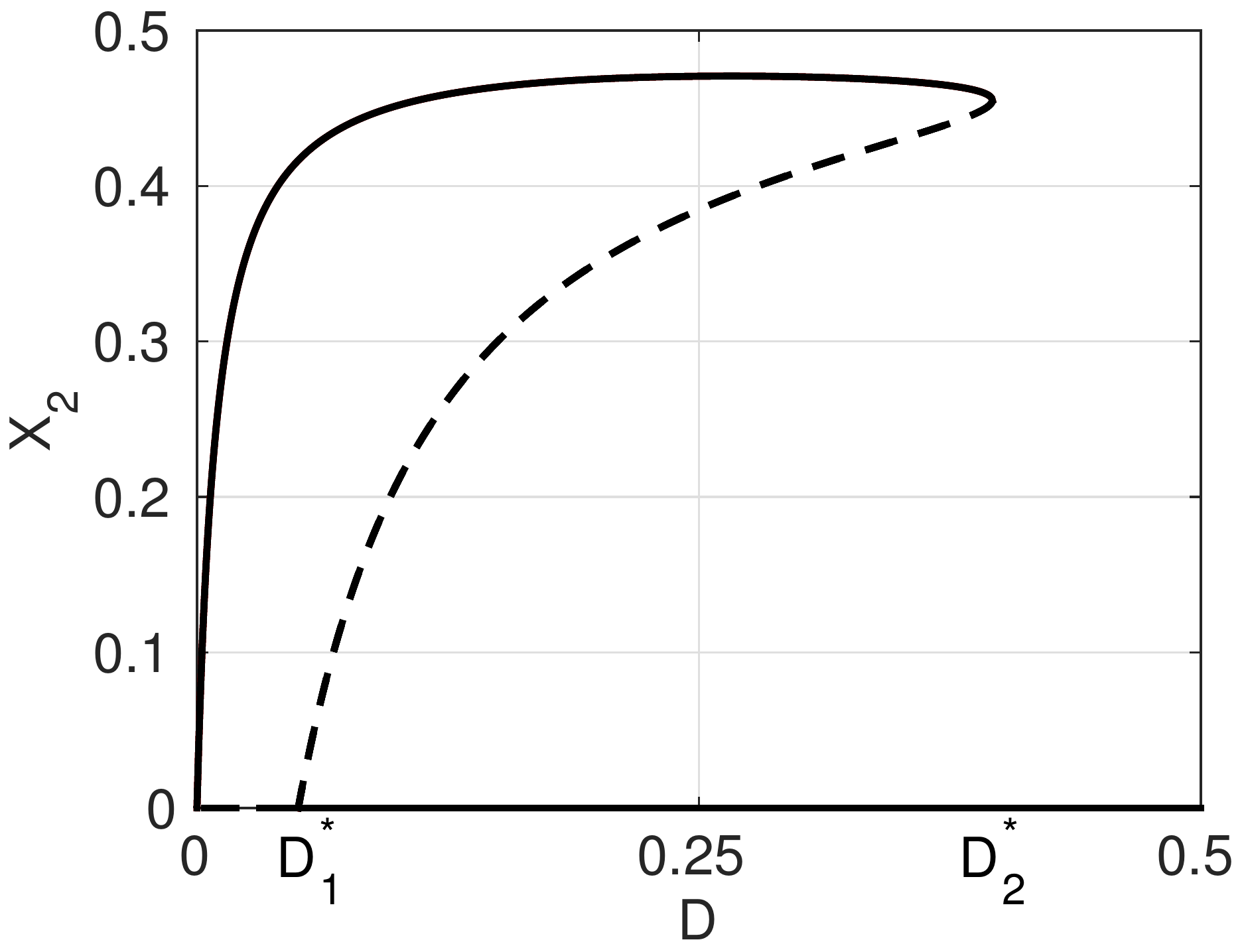}}
	\subfigure[\label{fig:bifd}]{\includegraphics[width=.3\textwidth]{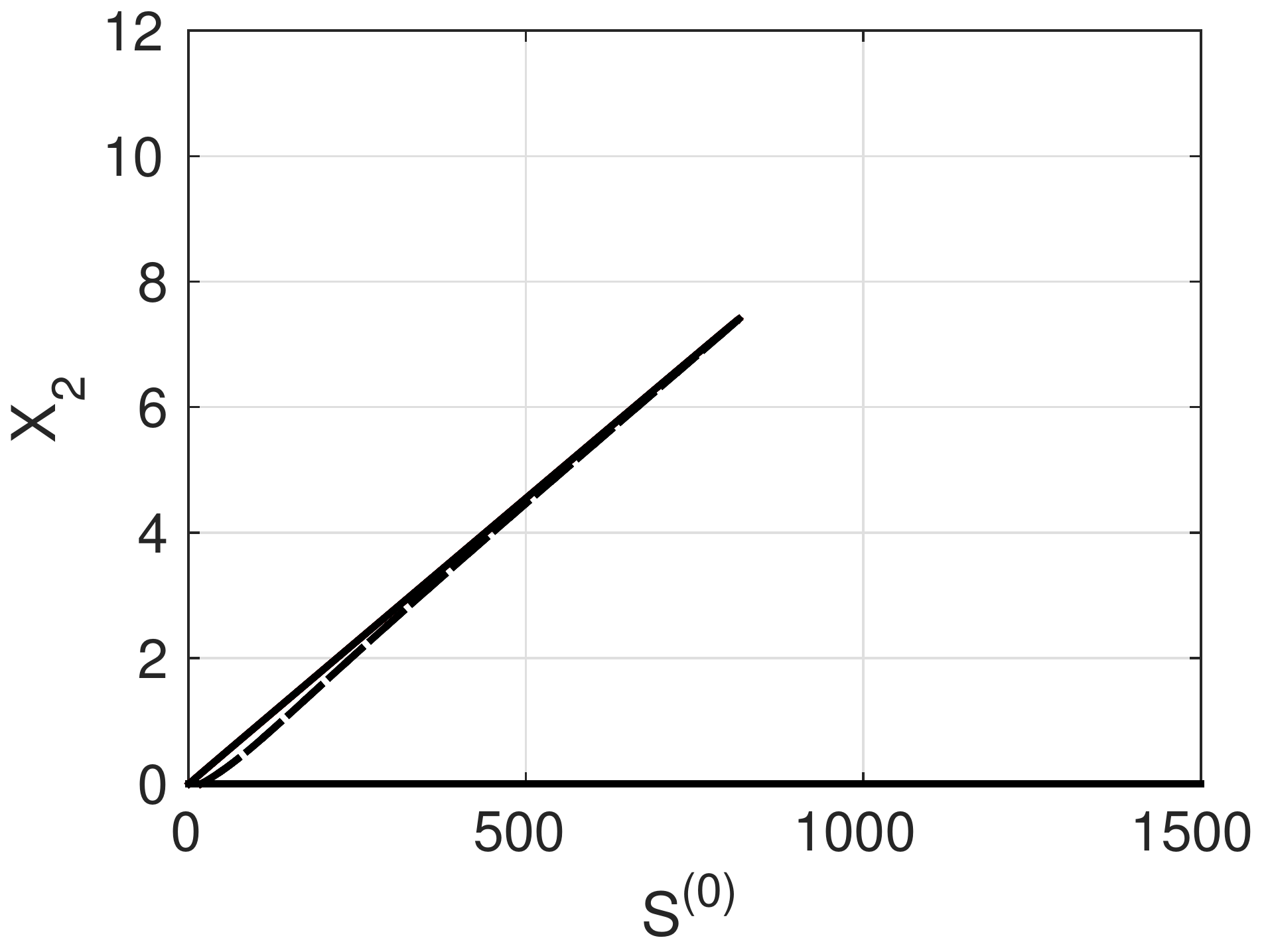}}
	\subfigure[\label{fig:bife}]{\includegraphics[width=.3\textwidth]{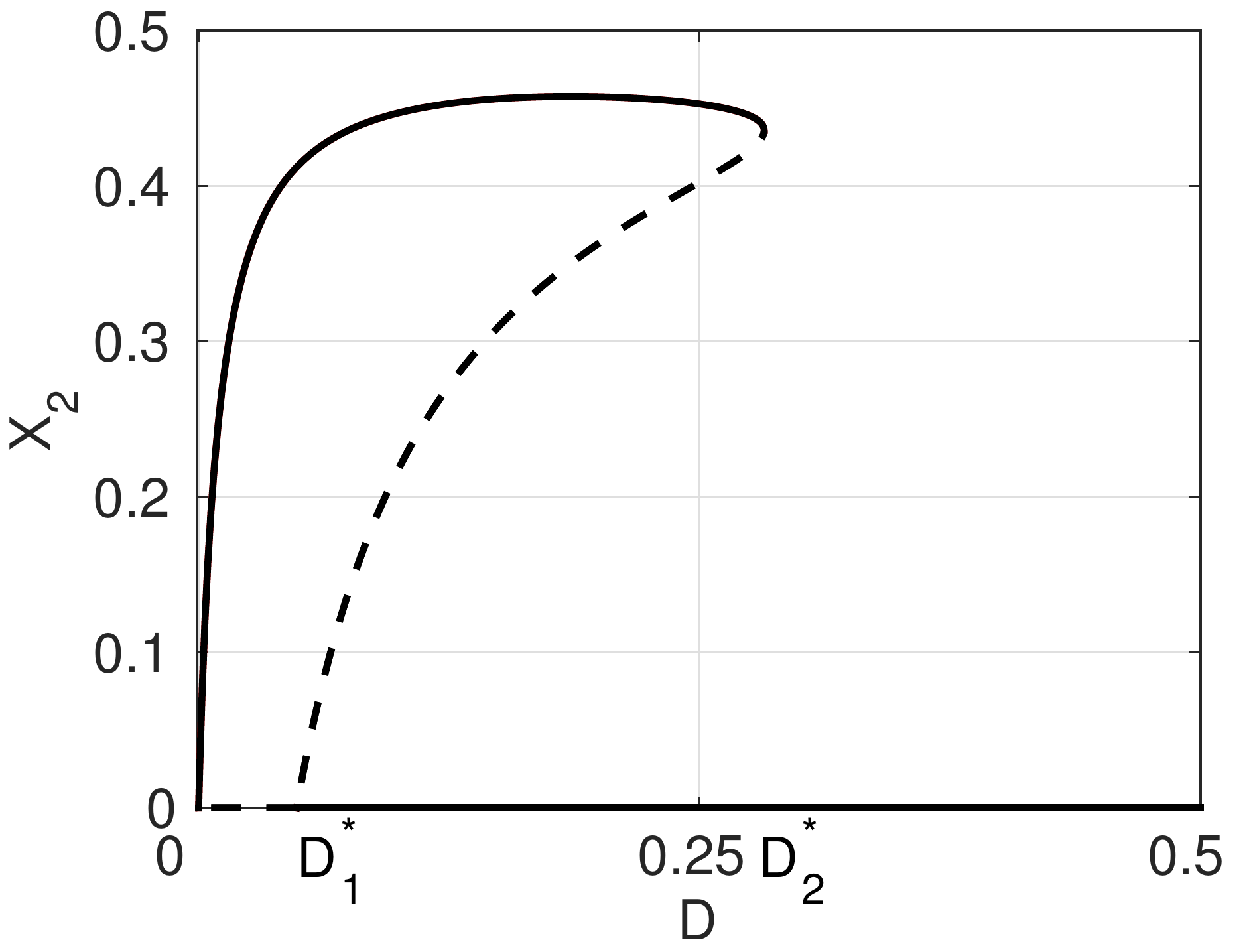}}
	\subfigure[\label{fig:biff}]{\includegraphics[width=.3\textwidth]{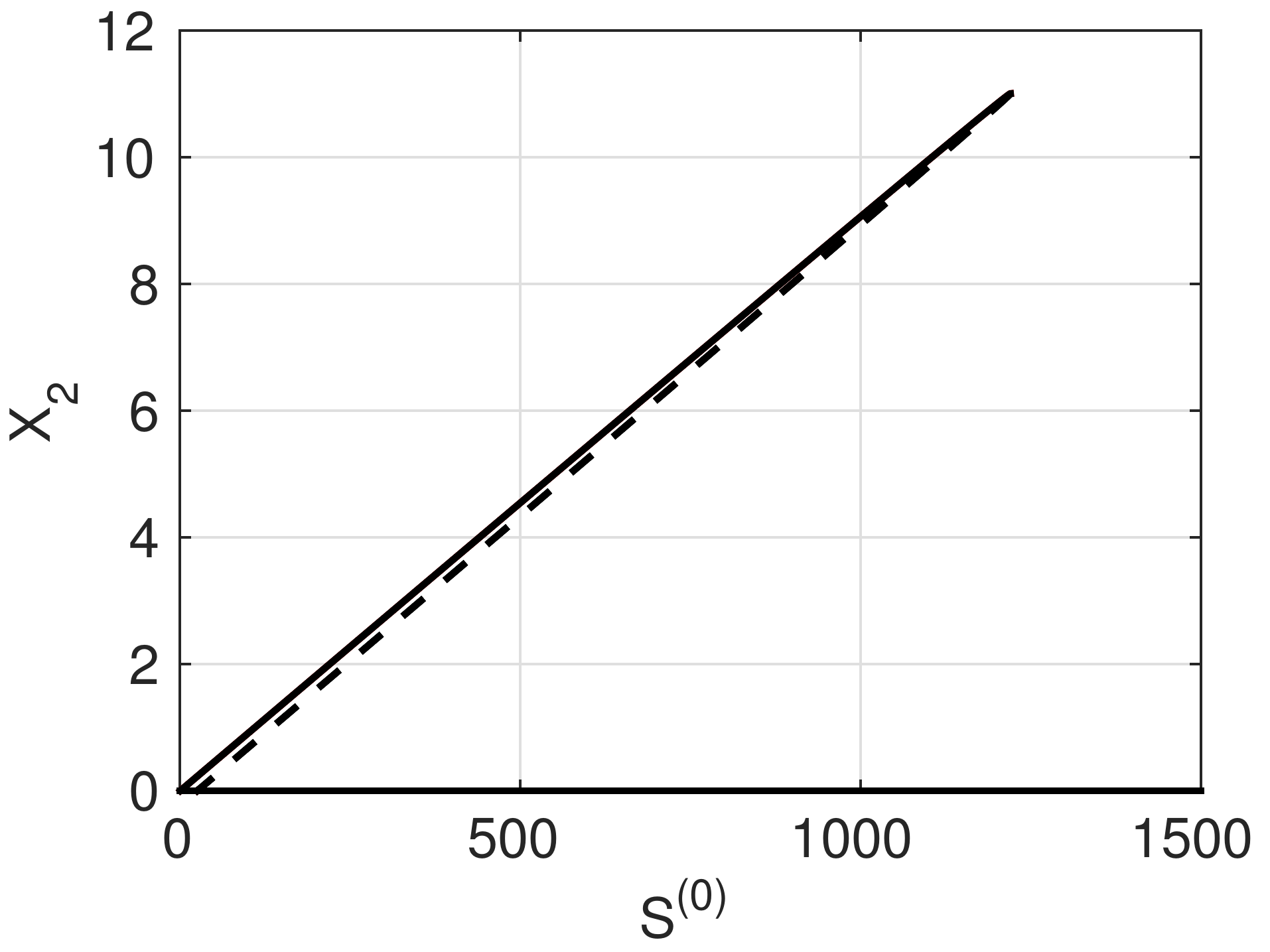}}
\caption{Bifurcation diagrams with bifurcation parameter $D$ in (a),
	(c), (e)  and $S^{(0)}$ in (b),(d), (f) and response function
	$\mu_{2}(S_2,S_3)=\mu_{2,I}(S_2,S_3)$ in (a) and (b), and
	$\mu_{2}(S_2,S_3)=\mu_{2,II}(S_2,S_3)$ in (c) and (d) and
	$\mu_{2}(S_2,S_3)=\mu_{2,III}(S_2,S_3)$ in (e) and (f).
	The solid curves correspond to the $X_2$ (methane
	producing)  coordinate of the
	asymptotically stable equilibrium of model \cref{eq:system}, and the dashed curves
	correspond to  the $X_2$ component of unstable equilibria.} \label{fig:bifurcation}
\end{figure*}
\begin{figure*}
\centering
	\subfigure[]{\includegraphics[width=.4\textwidth]{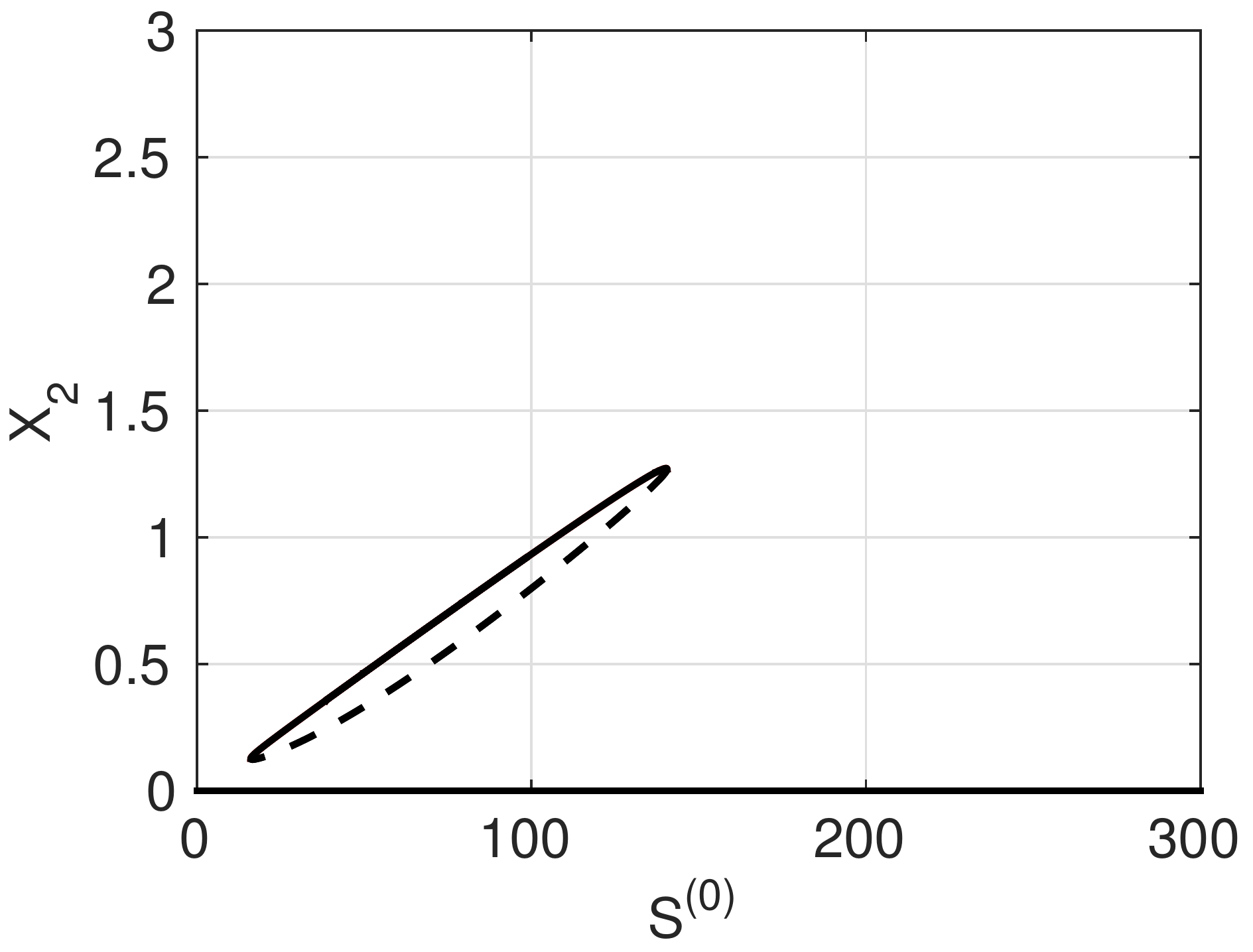}}
	\subfigure[]{\includegraphics[width=.4\textwidth]{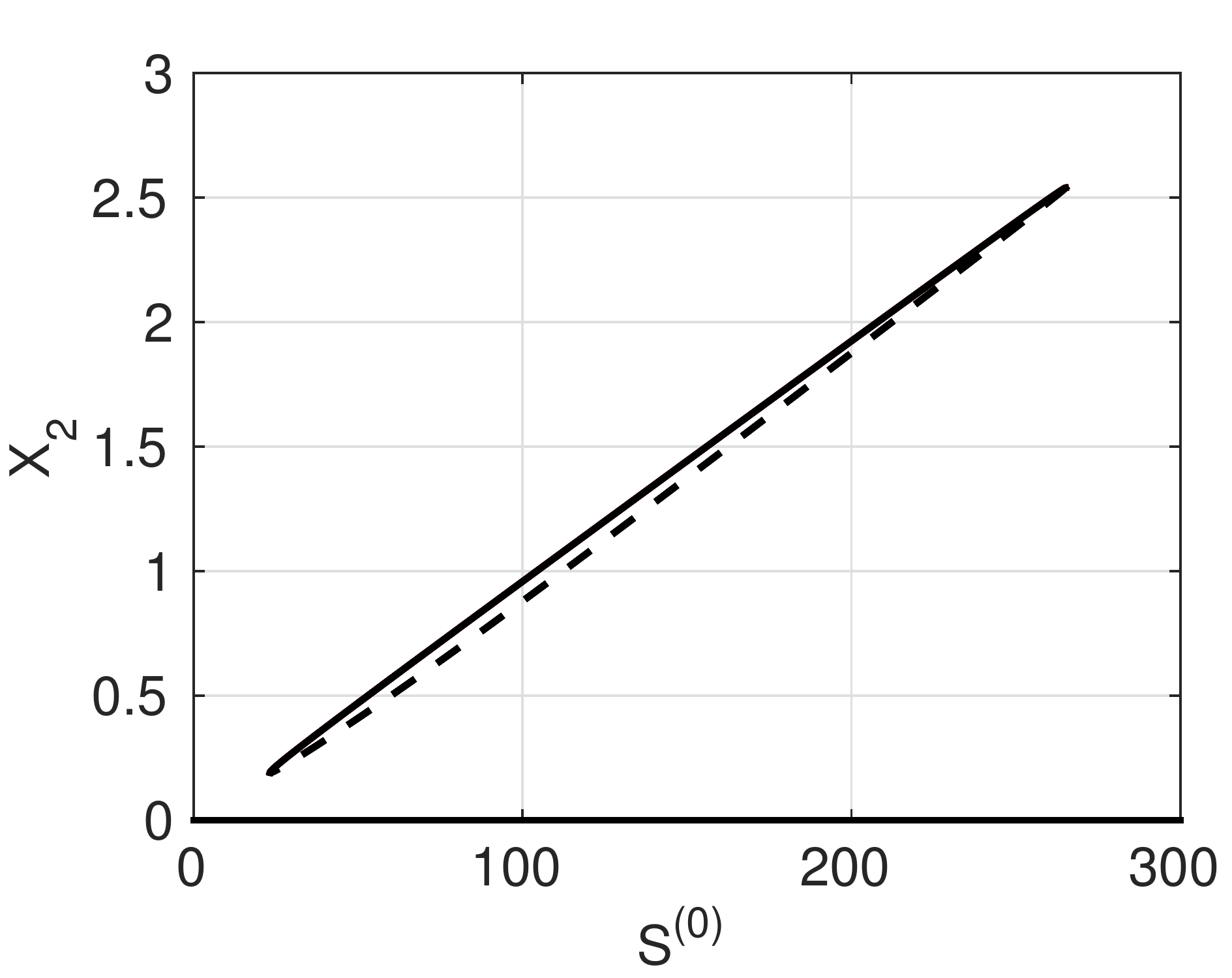}}
	\caption{Bifurcation diagrams with bifurcation parameter
  $S^{(0)}$ illustrating two saddle-node
	bifurcations.  (a)   $\mu_2(S_2,S_3) =\proto{I} (S_2,S_3)$.
	(b)  $\mu_2(S_2,S_3) =\proto{III}(S_2,S_3)$. The
	solid curves correspond to the $X_2$ component of  locally asymptotically stable equilibria, and
	the dashed curves correpond to the $X_2$ component of  unstable
	equilibria for model \cref{eq:system}. The
	analogous 
	diagrams for 
	$\proto{II}(S_2,S_3)$ do not exhibit this behaviour.}\label{fig:bifloop}
\end{figure*}

\section{Stochastic Simulations of the Full
System~\cref{eq:system} }\label{sec:stochastics}
We describe   two stochastic algorithms to capture
stochasticity  in the parameters.   For comparison we also include simulations done with Gillespie's stochastic simulation algorithm \cite{Gillespie:1977} and the adaptive tau-leaping algorithm \cite{Gillespie:2001}.

The simulations in this section are all done for the full system \cref{eq:system} with
\begin{equation}
\mu_1(S_1)=  \frac{\kappa S_1}{r_1+S_1} \nonumber
\end{equation} and $\mu_2(S_2,S_3) = \proto{III}(S_2,S_3)$.
 The parameters are listed in \cref{tab:parameters}. With these parameters, the
deterministic system  has two stable equilibria, $E_0$ and
$E_1$, and so the long-term behaviour of the solutions is initial condition dependent.  If
\begin{equation} \label{eq:IC_E1}
	(S_1(0),X_1(0),S_2(0),S_3(0),X_2(0))=(50,0.4,0,0,1.16),
\end{equation}
 the solution of the deterministic system converges to $E_1$ (see
\cref{tab:equilibria}), and if
\begin{equation} \label{eq:IC_E0}
	(S_1(0),X_1(0),S_2(0),S_3(0),X_2(0))=(50,0.4,0,0,1.14),
\end{equation}
 the solution of the deterministic system converges to $E_0$ (see
\cref{tab:equilibria}). Thus,  for one set of initial conditions, the
deterministic system \cref{eq:system} predicts that
  the methanogens survive  and produce biogas, and
    for the other it predicts that they do not.
  These initial conditions simulate the start
up and inoculation of the reactor.  The only difference
between  the  initial conditions in  \cref{eq:IC_E1,eq:IC_E0}
is the value of $X_2(0)$. Both initial conditions are close to the
separatrix.
We only include figures that show the population
of  methanogens, $X_2(t)$,
    to compare the effect of
stochasticity on biogas production, which  only occurs
if $X_2$ is positive.  In simulations (not shown) with initial
conditions   farther  from the separatrix, solutions converged to the
same equilibrium predicted by the deterministic model every time.
The figures were produced using Matlab \cite{MATLAB:2015}.

\begin{table}[tbhp]
{\footnotesize
\caption{Equilibria for   system \cref{eq:system} with parameters
	 given in \cref{tab:parameters}, with
	 $\mu_{2}(S_2,S_3) = \proto{III}(S_2,S_3)$.  \label{tab:equilibria}}
\begin{center}
\begin{tabular}{|l|l|}
	\hline
	\multicolumn{2}{|c|}{Equilibria} \\
\hline
$E$ & (50, 0, 0, 0, 0)\\
\hline
$E_0$ & (1.092, 1.088, 135.2, 1.352, 0)  \\
\hline
$E_1$ & (1.092, 1.088, 3.304, 1.352, 0.4614)\\
\hline
$E_2$ & (1.092, 1.088, 28.09, 1.352, 0.3747)\\
\hline
\end{tabular}
\end{center}
}	
\end{table}

 We use two different approaches to study the behavior of \cref{eq:system} under stochastic perturbations. The first method is meant to model fluctuations in the parameters due, for example, to fluctuations in the
environment. The second method captures the effect of potential mutations in members of the
populations. In both schemes, multiple parameters are perturbed at randomly chosen times. Because we are varying many parameters, some of which appear in the non-linearities of the system, we are unable to write the resulting stochastic equations as a linear stochastic perturbation of
the original system as was done in \cite{Wang:2016,Xu:2013} for chemostat models. In \cite{Wang:2016}, the dilution rate and in \cite{Xu:2013}, the dilution rate and the decay rates are assumed to vary stochastically.
In one
algorithm the perturbations are from the mean and in the other the
perturbations are accumulative.
Between perturbations the system is treated as a
deterministic system that is solved numerically.

Let $\tau_0=0$ and
$\tau_{i+1} = \tau_{i} -\ln(T_i)$, where $T_i\in(0,1)$ is a uniformly
distributed random variable. Therefore, $\{\tau_i\}$ describes a
monotone increasing sequence of times.  By applying the inverse sampling transform, we see that the difference $\tau_{i+1}-\tau_{i}$ is exponentially distributed with unit mean and variance. Let $P_0$ be a  row
vector  containing the parameter values present in
the  deterministic system that are  affected by
stochasticity.    At each  randomly chosen time  $\tau_i$,  these parameters values
are updated  to obtain a sequence of  vectors
$\{P_{\tau_i}\}_{i=1}^{\infty},$ and we set the
parameters  equal to
$P_t=P_{\tau_i}$, for $t \in[\tau_i,\tau_{i+1})$.

In the first stochastic  algorithm, which we call the
environmental  based fluctuation algorithm,  we assume that the parameter
values are influenced by the environment.  As such, they cannot be perfectly
controlled and so at random intervals of time they undergo small
random changes.  However, the parameters remain near their
mean values given in the row vector $P_0$.
Following this
interpretation, we let
$N_t$ be a diagonal matrix
with entries given by Gaussian random
variables with mean $\mu=1$ and standard deviation, $\sigma$.
 We assume that $N_t = N_{\tau_i}$ for $t\in[\tau_i,\tau_{i+1})$. Then
\begin{equation}
	P_{\tau_{i+1}} = P_0 N_{\tau_i}.
\end{equation}

\cref{fig:Stochastic_a,fig:Stochastic_b} show five
simulations using the
environmental  based algorithm   with
 $\sigma=\frac{1}{10}$ and
$$P_0 = [S_0,D,y_1,y_2,y_3,y_4,K,k_1,m_{II},r]$$
In \cref{fig:Stochastic_a} the initial conditions are
given by \cref{eq:IC_E0} and the solution to the
deterministic system converges to $E_0$. In  \cref{fig:Stochastic_b},  the initial conditions are
given by \cref{eq:IC_E1} and solutions converge to $E_1$. The solutions
for the deterministic system are shown in bold for comparison.

\begin{figure*}
\centering
	\subfigure[]{\includegraphics[width=.4\textwidth]{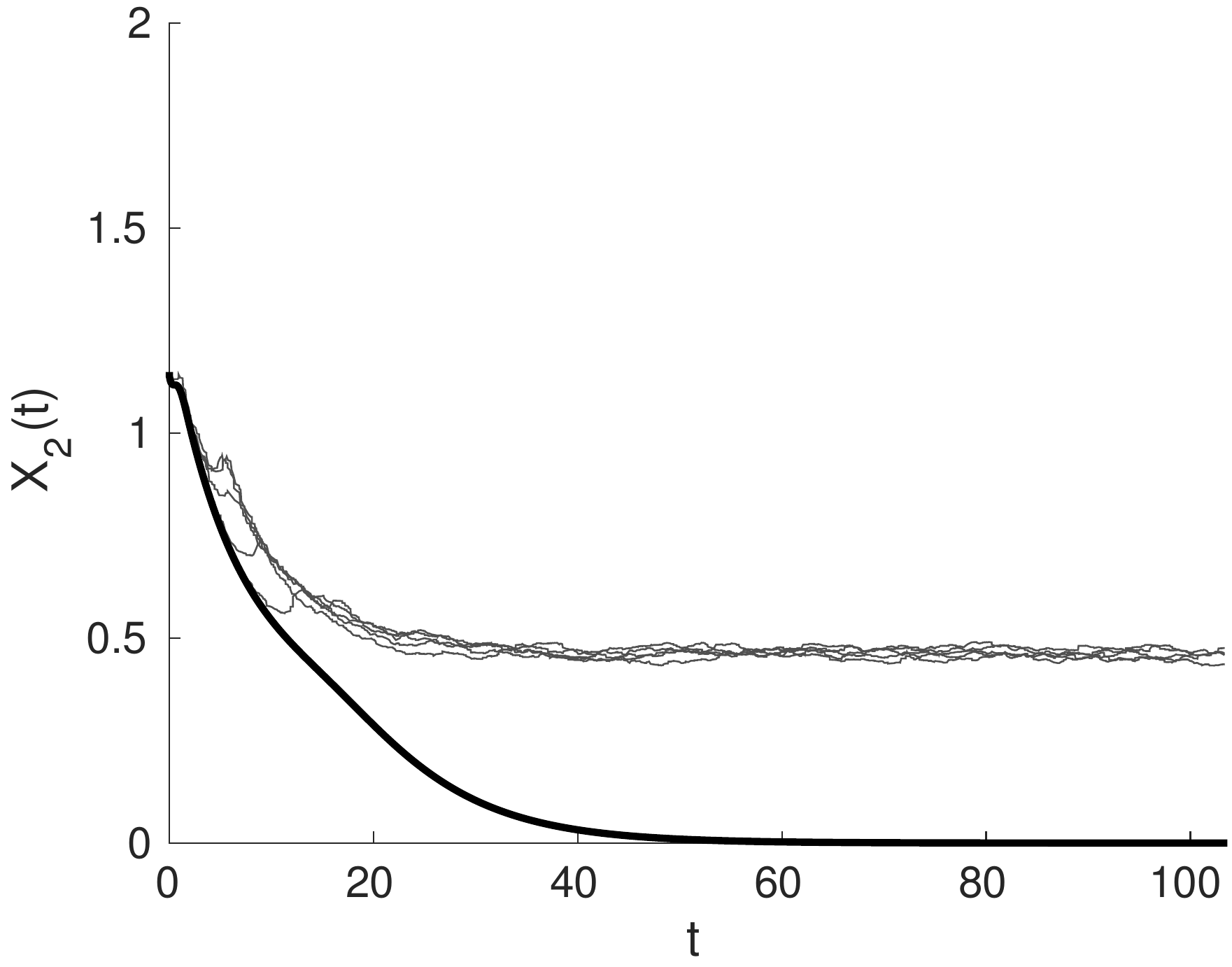}\label{fig:Stochastic_a}}
	\subfigure[]{\includegraphics[width=.4\textwidth]{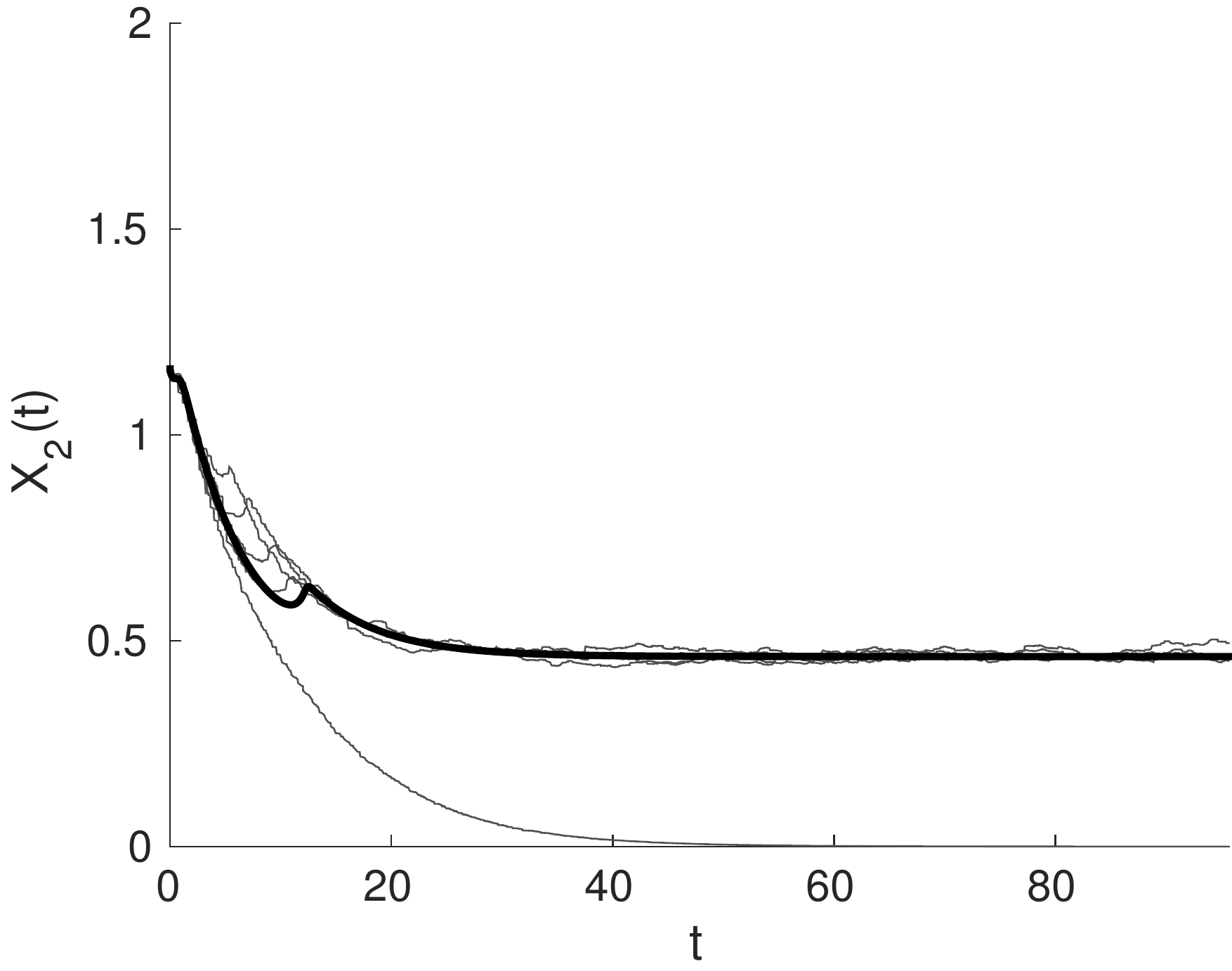}\label{fig:Stochastic_b}}
	\subfigure[]{\includegraphics[width=.4\textwidth]{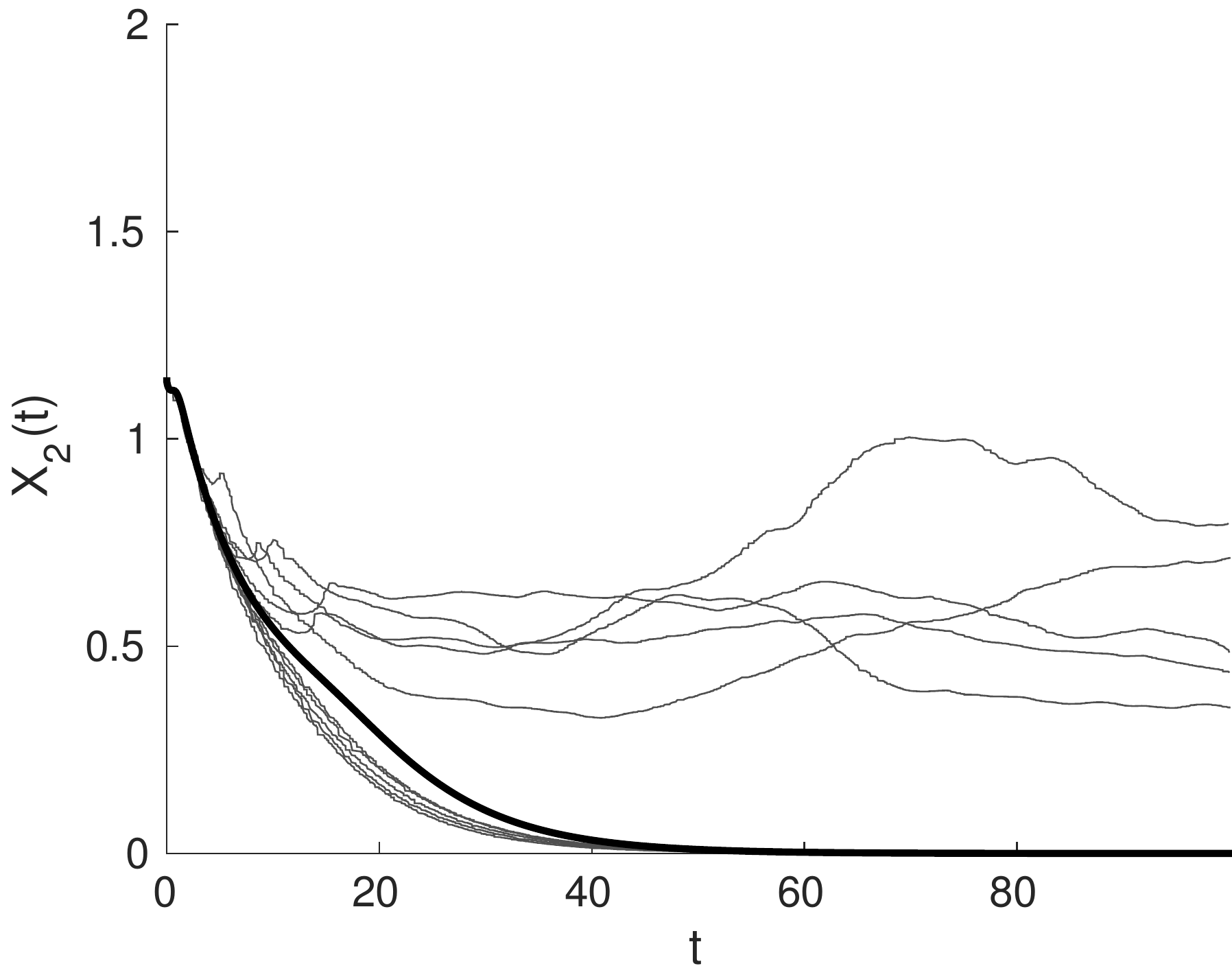}\label{fig:Stochastic_c}}
	\subfigure[]{\includegraphics[width=.4\textwidth]{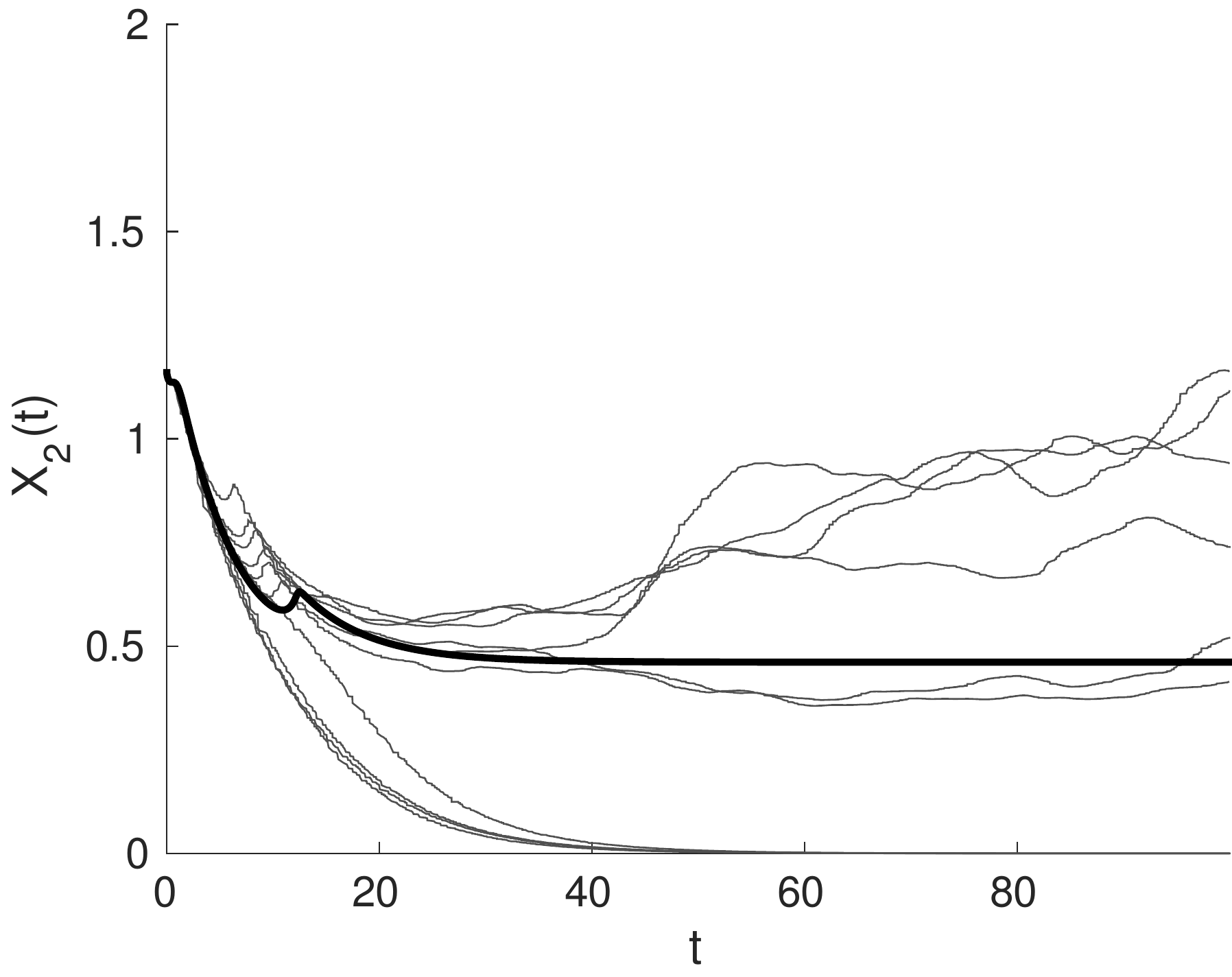}\label{fig:Stochastic_d}}
	\caption{
	  Sample paths of system \cref{eq:system}
	for the
	methanogens, $X_2(t)$, using the environmental fluctuation based method in \cref{fig:Stochastic_a,fig:Stochastic_b}, and using the mutation based method in \cref{fig:Stochastic_c,fig:Stochastic_d}. On the left, the initial
	conditions
	are given in
	\cref{eq:IC_E0} and are in the basin
	of attraction of $E_0$ for the deterministic system.   On the right, the initial conditions
	are given in
	\cref{eq:IC_E1} and are in the basin
	of attraction of $E_1$ for the deterministic system.
	The darker
	curve in each graph is the solution of the deterministic system
	and the lighter curves show the results of different stochastic
	runs. \label{fig:Stochastic_1}}
\end{figure*}
\begin{figure*}
\centering
	\subfigure[]{\includegraphics[width=.4\textwidth]{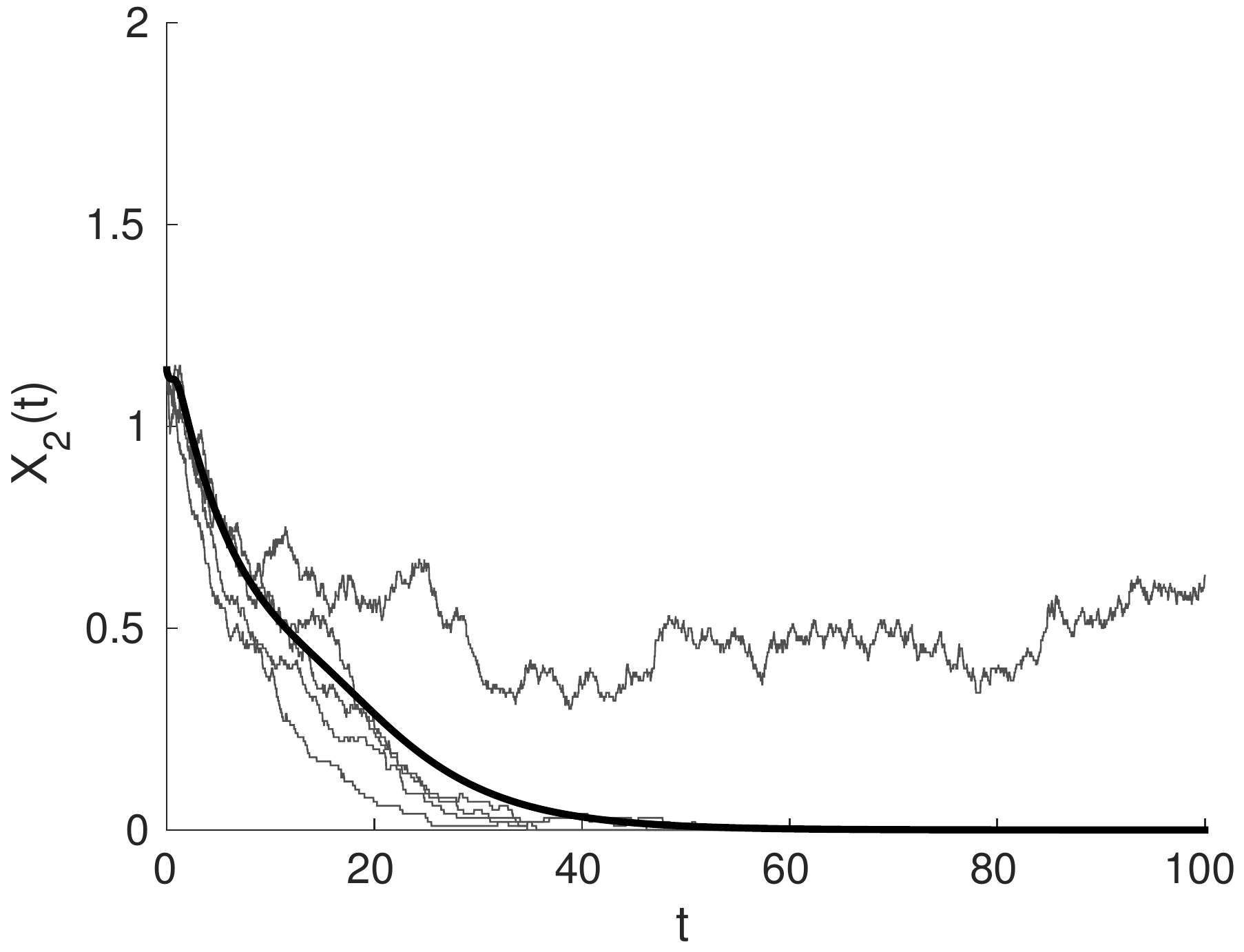}\label{fig:Stochastic_e}}
	\subfigure[]{\includegraphics[width=.4\textwidth]{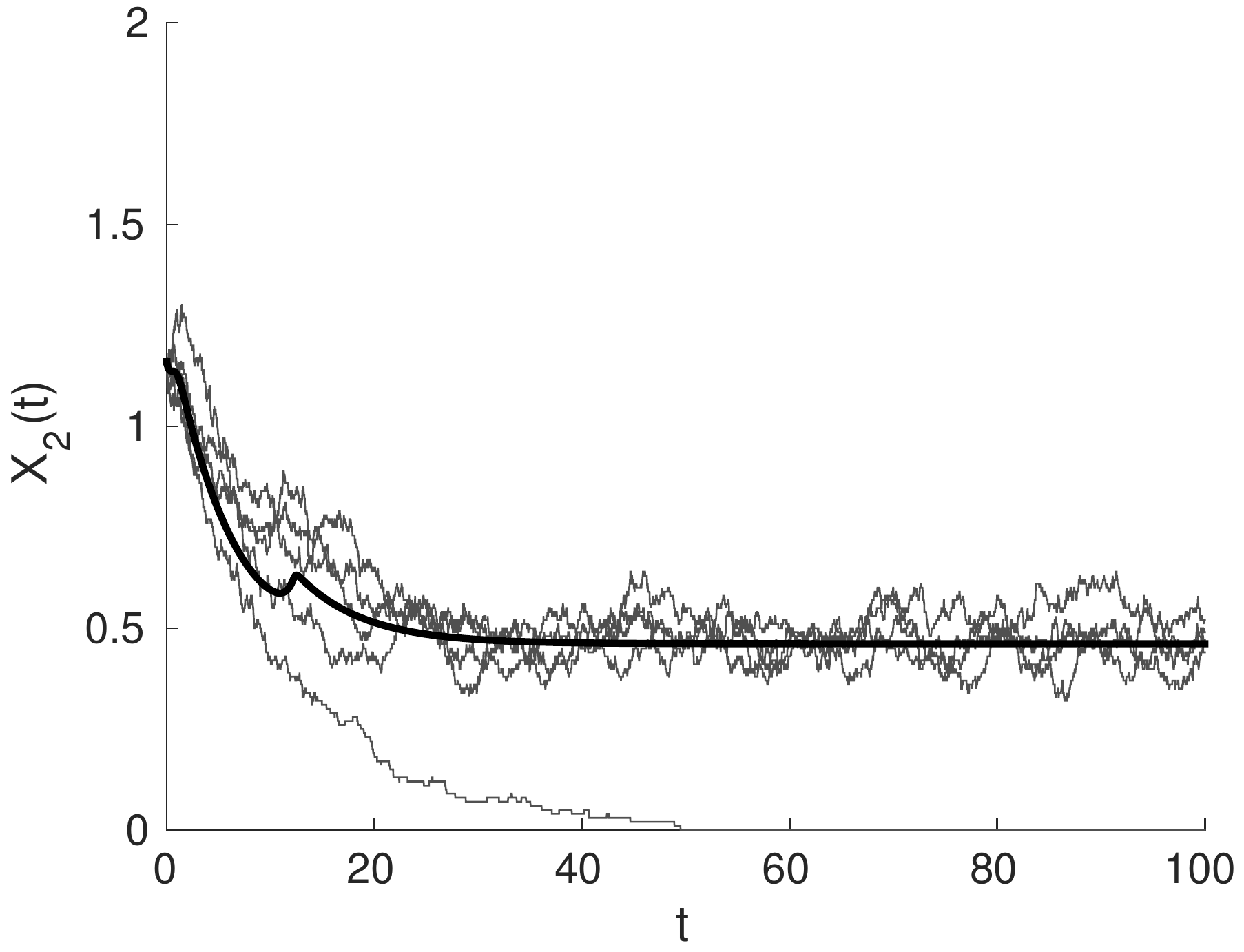}\label{fig:Stochastic_f}}
	\subfigure[]{\includegraphics[width=.4\textwidth]{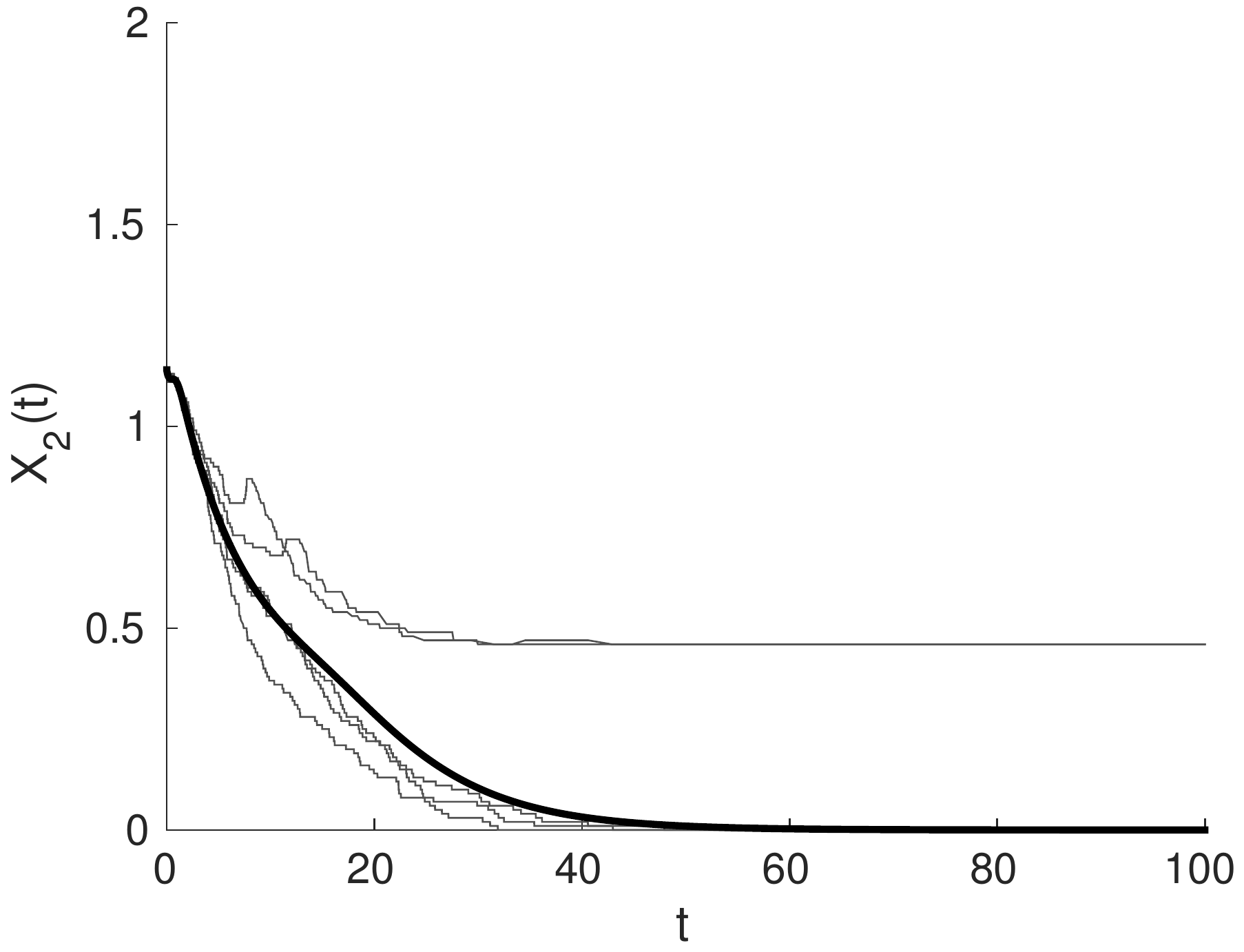}\label{fig:Stochastic_g}}
	\subfigure[]{\includegraphics[width=.4\textwidth]{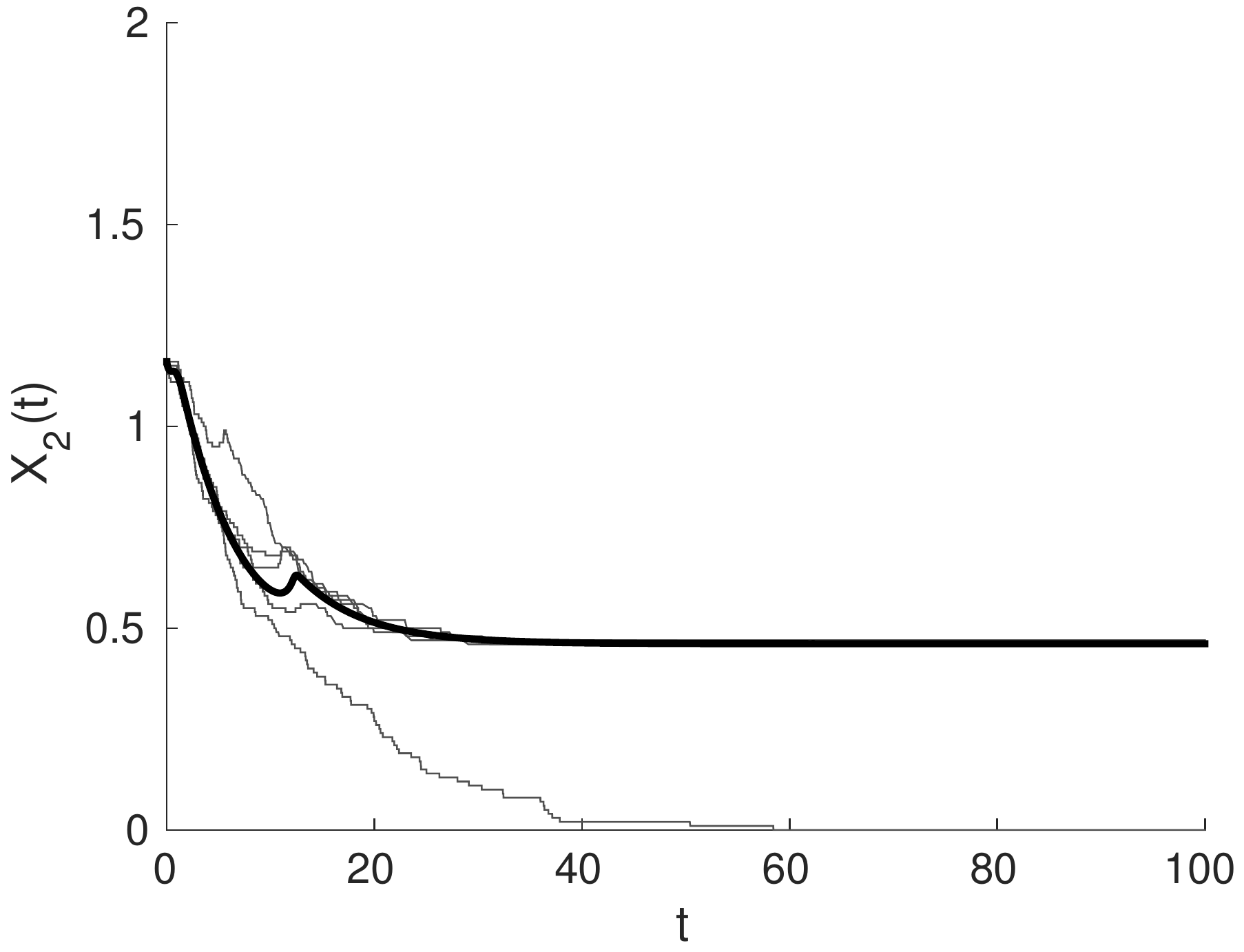}\label{fig:Stochastic_h}}
	\caption{Sample paths of system \cref{eq:system}
	for the
	methanogens, $X_2(t)$, using Gillespie's SSA in \cref{fig:Stochastic_e,fig:Stochastic_f}, and using the tau-leaping method in \cref{fig:Stochastic_g,fig:Stochastic_h}. On the left, the initial
	conditions
	are given in
	\cref{eq:IC_E0} and are in the basin
	of attraction of $E_0$ for the deterministic system.   On the right, the initial conditions
	are given in
	\cref{eq:IC_E1} and are in the basin
	of attraction of $E_1$ for the deterministic system.
	The darker
	curve in each graph is the solution of the deterministic system
	and the lighter curves show the results of different stochastic
	runs.}
\end{figure*}

In the second stochastic   algorithm, which we call the mutation based
algorithm,  we assume that the parameters are dependent on properties of the
microorganisms that can mutate, and therefore are subject to changes
  at random times that accumulate. In this case, many of the parameters are
beyond the control of the operator, however we assume that the operator
has  complete control of the dilution rate $D$ and the input concentration
$S_0$. Following this interpretation, we
  update the parameters at random times to obtain,
\begin{equation}
	P_{\tau_{i+1}} = P_{\tau_i}N_{\tau_i}= P_0\prod_{n=1}^iN_{\tau_n},
\end{equation}
where again $\sigma=\frac{1}{10}$,
$$ P_0=[y_1,y_2,y_3,y_4,K,k_1,m_{II},r],$$  and $N_{\tau_i}$ are
as before.   Using this algorithm,       $\{P_{\tau_i}\}_{i=1}^{\infty}$
is a random walk with mean $P_0$, and
the mutations accumulate.
Random walks have the property that $\sigma^2 \to \infty$ as $t\to
\infty$, and therefore the system is subject to wild fluctuations as time
increases. Care must be taken so that the parameters, which have
interpretations as positive quantities only, do not become negative. We
ensure non-negativity by taking $P_{\tau_{i+1}} =
\max\{0,P_{\tau_i}N_{\tau_i}\}$, and control the wild fluctuations by
limiting the difference between current parameter values $P_{\tau_i}$
and the initial parameter values $P_0$ to be less than four standard
deviations. 
\cref{fig:Stochastic_c,fig:Stochastic_d} shows five
simulations using the mutation
based algorithm.

We also include simulations using Gillespie's stochastic simulation
algorithm (SSA) \cite{Gillespie:1977}, in
\cref{fig:Stochastic_e,fig:Stochastic_f}. 
 The
SSA is an essentially exact description for systems with a finite number
of interacting particles. The SSA is based on the principle of mass
action, and as such the deterministic system must be converted to an equivalent
system that is of the form
\begin{subequations}\label{eq:massaction}
\begin{align}
\dot{S_1} &= \sum_{i,j,k,\ell,m} a_{ijklm}X_1^i S_1^jS_2^kS_3^\ell X_2^m,\\
\dot{X_1} &= \sum_{i,j,k,\ell,m} b_{ijklm}X_1^i S_1^jS_2^kS_3^\ell X_2^m,\\
\dot{S_2} &= \sum_{i,j,k,\ell,m} c_{ijklm}X_1^i S_1^jS_2^kS_3^\ell X_2^m,\\
\dot{S_3} &= \sum_{i,j,k,\ell,m} d_{ijklm}X_1^i S_1^jS_2^kS_3^\ell X_2^m,\\
\dot{X_2} &= \sum_{i,j,k,\ell,m} e_{ijklm}X_1^i S_1^jS_2^kS_3^\ell X_2^m.
\end{align}
\end{subequations}
To do so, we rescale the time variable by $dt
=(r_1+S_1)(K+k_1S_2+rS_2^2)(a+S_3^2)d\hat{t}$. The resulting system has
104 different reaction terms that  must be accounted for. As such,
reporting the system here would be impractical. Although we have rescaled
the time variable, the dynamics of system \cref{eq:massaction} are
identical to those of \cref{eq:system}. The SSA assumes that each
reaction occurs independent of the others, and occurs with rates given
by the coefficients of the differential equations. The SSA determines a
time until each reaction takes place using the rate coefficients and the
population of individuals relevant to that reaction, and increases or
decreases the population(s) of the fastest reaction by a set step size.
Once we have realized the simulation, we scale time back to the original
time variable before plotting in order to compare with the other
stochastic   algorithms. Five simulations with a step size of
$\tfrac{1}{100}$ are shown in \cref{fig:Stochastic_e,fig:Stochastic_f}. In reality, the
step size is meant to represent a single individual in the population,
but since SSA is notoriously slow, modelling a population of trillions
of microorganisms and on the order of $10^{23}$ molecules in this way is computationally impossible. It is also well known that as you decrease the step size, the SSA will approach the deterministic solution \cite{Gillespie:1977}.

Finally, we include simulations using the adaptive tau-leaping algorithm
in \cref{fig:Stochastic_g,fig:Stochastic_h}. 
The tau-leaping algorithm is an improvement on the SSA in terms of
speed, and is generally easier to implement, although it is less
accurate. One interpretation of the tau-leaping algorithm is that it is analogous to Euler's method, but instead of the derivative,  a Poisson random variable with mean proportional to the derivative is used.
 Here, \cref{eq:system} takes the form
\begin{subequations}
\begin{align}
 S_1(t+\tau) &=  S_1(t) + \delta P(\tau \dot{S_1}(t)),\\
 X_1(t+\tau) &=  X_1(t) + \delta P(\tau \dot{X_1}(t)),\\
 S_2(t+\tau) &=  S_2(t) + \delta P(\tau \dot{S_2}(t)),\\
 S_3(t+\tau) &=  S_3(t) + \delta P(\tau \dot{S_3}(t)),\\
 X_2(t+\tau) &=  X_2(t) + \delta P(\tau \dot{X_2}(t)),
\end{align}
\end{subequations}
where $\delta$ is the step size (typically interpreted to be an individual particle, as with the SSA). There has been much discussion on how to choose $\tau$ appropriately \cite{Gillespie:2006,Gillespie:2001}. We chose
\begin{equation}
\tau= \min\left\{\frac{1}{|\dot{S_1}|},\frac{1}{|\dot{X_1}|},\frac{1}{|\dot{S_2}|},\frac{1}{|\dot{S_3}|},\frac{1}{|\dot{X_2}|}\right\}
\end{equation}
so that the fastest reaction determines $\tau$.

The stochasticity as simulated in the environmental based fluctuation algorithm and the mutation based algorithm stems from
uncertainty in the system parameters, whether due to environmental noise
or from mutations. The stochasticity of the SSA and tau-leaping
algorithm is derived from the fact that the populations are treated as
discrete quantities. Since the populations are very large
in practice, it may be more realistic to implement stochasticity using
continuous hybrid algorithms that reflect the uncertainty in the parameters.

In  the simulations using all four  algorithms, if
the stochasticity caused the system
to  predict a different
outcome than the deterministic  system, it usually happened while the system was transient. Once the system neared an
equilibrium, the behaviour was usually quite stable. In rare instances,
noise caused the system to destabilize after nearing an equilibrium, but
this seemed only to occur for the 
mutation based method
when the noise was quite large.

\section{Conclusion}\label{sec:conclusion}

We analyze the system introduced by Bornh\"oft et al.
\cite{Bornhoft:2013}, which was proposed as a qualitative reduction of
the ADM1 model, and claimed to capture the most relevant qualitative
features of the ADM1 model. We give a complete global analysis of the
dynamics of the model. If the concentration of the simple substrates is
too low, both the acidogenic and methanogenic populations of
microogranisms are eliminated from the reactor and no biogas is
produced. Even if the input concentration of simple substrates is high
enough, if the equilibrium concentration of VFAs produced by the
acidogenic microorganisms is too low, then the methanogenic
microorganisms will be eliminated from the reactor, and the system will
converge to an equilibrium where no biogas is produced. If the VFA
concentration is in a proper range, the system has a single globally
stable interior equilibrium. Finally, if the equilibrium concentration
of VFAs is very high, then the system possesses two stable equilibria
and one unstable equilibrium, and no sustained oscillatory behaviour is
possible. In this case the long-term behavior is initial condition
dependent. Only one of the two stable equilibria corresponds to the
production of biogas meaning that it depends on the initial conditions
whether the reactor will produce biogas in the long-term. The system
does not allow bistability involving two or more biogas producing
equilibria, previously shown to be possible for the ADM1 model \cite{ADM1} and for the models studied in \cite{WeeSeoWolk:2013,WeeWolSas:2015}

The dynamics predicted by a bifurcation analysis of the model is
qualitatively similar for all three prototype functions. Ammonia
inhibition is included in the ADM1 model, however, in ADM1 ammonia is
not included as a dynamic variable. Ammonia concentration in ADM1 is
computed as the difference of the concentration of inorganic nitrogen
and $NH_4^+$. In the present model, ammonia is included as a dynamic
variable and it is important to determine how to best model the effect
of ammonia on the growth of the methanogens to capture the behaviour of
ADM1. For all three prototype functions, inhibition of the growth of
acetoclastic methanogens due to ammonia is unimodal with respect to the
ammonia concentration. However, for $\proto{I}(S_2,S_3)$, acetoclastic
methanogens will not grow in the absence of ammonia, while for
$\proto{II}(S_2,S_3)$ and $\proto{III}(S_2,S_3)$ the organisms grow even
if the ammonia concentration is zero.   Based on a
comparison with Fig.~10 in \cite{Bornhoft:2013}, using  $\mu_{2,II}$ or
$\mu_{2,III}$ in model  \cref{eq:system},   the
behavior resembles the behavior or the   ADM1 model shown in
\cite{Bornhoft:2013} more closely than using $\proto{I}(S_2,S_3)$. This indicates that these two functions are better suited to model the dependence of acetoclastic methanogens on ammonia.

We consider two algorithms that simulate stochastic   effects in system
\cref{eq:system}. The aim of these two algorithms is to  model the
uncertainty and variation in environmental and biological parameters
that are  hard to control with  numerical algorithms that are easy to implement and run
relatively quickly.
We compare the resulting graphs  with the graphs
produced using  the well-known  Gillespie algorithm and the
the tau-leaping algorithm.The
stochastic simulations from all four algorithms seem to  indicate that a failure of the reactor is most
likely to occur early in the reactors operating cycle, and that once the
reactor has reached  a steady state, it is quite
resilient and less affected by minor perturbations due to mutations or
small fluctuations in the environment.
The one possible exception is in  our mutation
based stochastic algorithm  that is intended to
simulate the accumulation of mutations within the
 microbial population. Therefore, it appears to be most important to control  the environment of the
reactor during start up, and then to carefully monitor
 the characteristics of the microorganisms within the reactor after start up.

 The analysis of the model of anaerobic digestion proposed
by Bornh\"oft et al.
\cite{Bornhoft:2013}   involved studying the  limiting system
\cref{eq:Chemostat}, a model of growth in the chemostat in the case of
a non-monotone response function with species decay rate added to the dilution rate.
 Armstrong and McGehee \cite{MA:1977} considered   model
 \cref{eq:Chemostat} extended to $n$ species competition  in the case
 of  monotone   response functions. By ignoring the species decay rate, they were able to apply a conservation law to obtain a limiting system.
They then studied the resulting limiting system,
but did not apply the theory of
asymptotically autonomous systems to obtain results for the full
system.  Butler and Wolkowicz \cite{BW:1985} used 
a different method, provided a complete global
analysis of this $n$ species model for both arbitrary  monotone and
non-monotone response
functions, and  applied results for asymptotically autonomous systems  so
that their results applied to the full system, not just the limiting system.
They proved that competitive exclusion holds, i.e., all solutions
approach an equilibrium that can be initial condition dependent in
the non-montone case.
In  Wolkowicz and Lu \cite{WL:1992}, the decay rates were no longer
ignored. There it
was proved that for a large class of monotone and non-monotone response
functions, 
again competitive exclusion holds and all populations
approach equilibrium. However, in the case of  non-monotone response
they only considered  the case  when the species with the lowest break-even
concentration also has its larger break-even concentration larger than
the substrate input concentration.
 In the case of only one species, their method works for all monotone
 response functions, but for
non-monotone response functions still requires the assumption that   the larger break-even concentration
is larger than the input concentration.  In this paper we were able to
eliminate this assumption,
and hence complete the analysis for the model of growth in the basic
chemostat.

\begin{appendix}\section{Proofs}\label{ap:quasiproof}

	\begin{proof} {\it of \cref{prop:properties}}

		\textit{i)}   Assume
	first that $X_1(0)=0$ and all other initial conditions are
	non-negative. It follows that $X_1(t)= 0$, for all $t\geq 0 $. Hence, \cref{eq:system} reduces to the system of first order differential equations
\begin{subequations}
	\begin{align}
		\dot{S_1} &= (S^{(0)}-S_1)D,\label{eq:reducedsubstrates}\\
		\dot{S_2} &= -DS_2-y_3\mu_2(S_2,S_3)X_2,\label{eq:reducedacid}\\
		\dot{S_3} &= -DS_3,\label{eq:reducedammonia}\\
		\dot{X_2} &= -D_2X_2 + \mu_2(S_2,S_3)X_2\label{eq:reducedmethanogens}.
	\end{align}
\end{subequations}

Equations \cref{eq:reducedsubstrates} and \cref{eq:reducedammonia}
		imply that $S_1$ and $S_3$ converge exponentially to
		$S^{(0)}$ and $0$, respectively. The hyperplane given by
		$S_2=0$ is invariant under \cref{eq:reducedacid} by
		(H6), and the hyperplane given by $X_2=0$ is invariant
		under \cref{eq:reducedmethanogens}. By uniqueness of
		solutions to initial value problems, if $S_2(0)\geq0$
		and $X_2(0)\geq 0 $, then   $S_2(t)\geq 0$ and
		$X_2(t)\geq 0$ for all $t\geq0$. Consider $\Sigma =
		S_2+y_3X_2$. Then $\dot{\Sigma} = -DS_2-y_3D_2X_2 \leq
		-D \Sigma$ and thus $\Sigma(t)\to 0$ as $t \to 0$, implying $X_2(t)$ and $S_2(t)$ must each converge to 0 as $t \to \infty$.

	\textit{ ii) and iii)} Assume that $X_1(0)>0$ and $X_2(0)\geq0$,
	with all other initial conditions non-negative. Notice first
	that \cref{eq:substrates} and \cref{eq:acidogens} decouple
	from the system. They describe a simple chemostat, for which it
	is  known that if $X_1(0)>0$ and $S_1(0)\geq 0$, then
	$S_1(t)>0$ and $X_1(t)>0$ for all $t>0$ (e.g., see
	\cite{SW:1995,WL:1992,Harmand:2017}).
	Note that the hyperplane $X_2 = 0 $ is invariant  under
	\cref{eq:reducedmethanogens}, and so if $X_2(0)=0$,   $X_2(t)
	= 0$ for all $t\geq 0$, and if $X_2(0)>0$, then $X_2(t) > 0$ for
		all $t\geq 0$. If $S_3(0)=0,$ then by
		\cref{eq:ammonia}, $\dot{S}_3(0)>0$, and so there exists
		$\epsilon>0$ such that $S(t)>0$ for  all
		$t\in(0,\epsilon)$. Let
		$S_3(0) \geq 0$.  Suppose that there exists  $\hat{t}>0$ such that
		$S_3(t) >0$ for  all $t\in (0,\hat{t})$ and $S_3(\hat{t}) =
		0$. Then, $\dot{S}_3(\hat{t})\leq
		0$, but again from \cref{eq:ammonia},
		$\dot{S_3}(\hat{t}) = y_4\mu_1(S_1(\hat{t}))X_1(\hat{t})
		> 0$, a contradiction.  Hence, $S_3(t)>0$ for all $t>0$.
		Using \cref{eq:aceticacid}, a
		similar argument applies to $S_2$.

	\textit{iv)} It is  known (e.g., see
	\cite{SW:1995, WL:1992, Harmand:2017})
	that solutions to the simple chemostat
	(\cref{eq:substrates,eq:acidogens}) are   bounded.
	Hence, there exists $0<M<\infty$ such that  $S_1(t)<M$ and
	$X_1(t)<M$ for all $t\geq 0$ . Thus, $S_3$ satisfies
\begin{equation}\label{eq:diffineq1}
	\dot{S_3} \leq -DS_3+y_4\tilde{M},
\end{equation}
where $\tilde{M} = \mu_1(M)M$. This differential inequality implies that $S_3(t)\leq \frac{y_4\tilde{M}}{D}+S_3(0)e^{-Dt}$ for all $t\geq 0$, and thus $S_3(t)$ is bounded for $t>0$. Since $D_i\geq D$ the following differential inequality holds
	\begin{equation}\label{eq:diffineq}
	\dot{X_2}\leq -DX_2+\mu_2(S_2,S_3)X_2.
	\end{equation}
Let $\Sigma = y_3X_2+S_2-\frac{y_2}{y_4}S_3$. Using
	\cref{eq:diffineq}, we see that $\dot{\Sigma} \leq -D\Sigma$, which implies that $\Sigma(t) \leq \Sigma(0)e^{-Dt}$. Since $S_3(t)$ is bounded above and we know that $S_2(t), X_2(t) \geq 0$ for all $t\geq 0 $, they too must be bounded above.
\qed
\end{proof}

	\begin{proof} {\it of \cref{prop:washout}}

		  Since \cref{eq:substrates,eq:acidogens} depend only on
		 	$S_1(t)$ and $X_{1}(t)$, these equations decouple from the full
		system \cref{eq:system}, and it follows
		from known results on the basic model of the
		chemostat (e.g., see
		\cite{SW:1995,Harmand:2017} )
		that if $\lambda_1\geq S^{(0)}$, then $(S_1(t),X_1(t)) \to (S^{(0)},0)$  as $t\to\infty$.

	Therefore,	for any $\epsilon>0$, there is a $T>0$ such
            that for $t>T$, $S_1(t) < S^{(0)} + \epsilon$ and $X_1(t) <
		\epsilon$. Then, for $t>T$,
            $\dot{S_3}(t)\leq -DS_3(t) +y_4 \mu_1(S^{(0)}+\epsilon)\epsilon$, which gives
            $S_3(t) \leq S_3(T) e^{-Dt} + \frac{y_4}{D}\mu_1(S^{(0)}+\epsilon)\epsilon \left( 1 - e^{-D(t-T)} \right)$.
            Then, $\lim_{t\to\infty} S_3(t) =
		\frac{y_4}{D}\mu_1(S^{(0)}+\epsilon)\epsilon $. Since
		this holds for all $\epsilon >0$, letting $\epsilon\to
		0$, gives $\lim_{t\to\infty} S_3(t) = 0$.
        Next, let $\Sigma_2(t) = S_2(t)+\frac{1}{y_3} X_2(t)$. Since
		$D\leq D_2$, $\dot{\Sigma_2}(t) \leq -D \Sigma_2(t) +
		y_2 \mu_1(S_1(t))X_1(t)$. The same argument as before
		proves that $\lim_{t\to\infty} \Sigma_2(t) = 0$. Since
		for all $t$, $S_2(t)\geq 0$ and $X_2(t)\geq 0$, it
		follows that
        $\lim_{t\to\infty} S_2(t) = \lim_{t\to\infty} X_2(t) = 0$.
\end{proof}

In the proof of \cref{thm:fullsystem} we rely on the fact that \cref{eq:system} is quasi-autonomous with the limiting system \cref{eq:Chemostat}. For completeness, we include both the definition of quasi-autonomous, and Theorem~1.4 of \cite{Th:1994}.

\begin{definition}\label{def:quasiautonomous}
Let $X$ be an open subset of $\mathbb{R}^k$, $k \geq 1$.  A system 
\begin{equation}
	\dot{x}(t) = f(t,x(t))
\end{equation}
with $x(t) \in X$ is called a \emph{quasi-autonomous system} with limiting system
\begin{equation}
	\dot{y}(t) = g(y(t))
\end{equation}
if for any compact set $K\subset X$
\begin{equation}
	\int_{t_0}^\infty \sup_{x(t)\in K} ||f(t,x(t))-g(x(t))||dt < \infty.
\end{equation}
\end{definition}
\begin{thm}[H. Thieme]
Let 
\begin{equation}\label{eq:AA}
\dot{x} = f(t,x)
\end{equation}
be quasi-autonomous with limit system
\begin{equation}\label{eq:L}
\dot{y} = g(y).
\end{equation}
Let $x$ be a forward bounded solution of \cref{eq:AA} that is defined for all forward times such that the closure of its forward orbit is contained in the open set $X\subset \mathbb{R}^2$. The following alternative holds for the $\omega$-limit set of 
$x$, $\omega$:
\begin{enumerate}
\item $\omega$ contains a periodic orbit of \cref{eq:L}.
\item $\omega$ contains at least one equilibrium of \cref{eq:L} and no periodic orbits of \cref{eq:L}.
\end{enumerate}
\end{thm}

To show that system~\cref{eq:system} is a quasi-autonomous system with
limiting system \cref{eq:limitsystem}, we first prove a lemma.

\begin{lem}\label{lem:quasi-autonomous}
Let $\dot{x}(t){=}f(t,x(t))$ be quasi-autonomous with limiting system $\dot{y}(t){=}$ $g(y(t))$ and assume that there exists $h(x(t))$ such that for all $K\subset X$ compact
\begin{equation}
	\int_{t_0}^\infty \sup_{x(t)\in K} ||g(x(t)) - h(x(t))||dt < \infty.
\end{equation}
Then $\dot{x}(t) = f(t,x(t))$ is quasi-autonomous with limiting system $\dot{y}(t) = h(y(t))$.
\end{lem}

\begin{proof}
	 By the triangle inequality, 
	\begin{align*}
		\int_{t_0}^\infty \sup_{x\in K} ||f(t,x)-h(x)|| dt &\leq \int_{t_0}^\infty \sup_{x\in K} ||f(t,x)-h(x)+g(x)-g(x)|| dt \\
 & \leq \int_{t_0}^\infty \sup_{x\in K} ||f(t,x)-g(x)||+||g(x)-h(x)|| dt\\
	&\leq \int_{t_0}^\infty \sup_{x\in K} ||f(t,x)-g(x)|| dt+\int_{t_0}^\infty \sup_{x\in K} ||g(x)-h(x)|| dt\\
	&<\infty .
	\end{align*}
\qed
\end{proof}

	\begin{proof} {\it of \cref{prop:limit}}
				
		First we show  that
		\cref{eq:system} is quasi-autonomous with limiting
		system:
\begin{subequations}\label{eq:intermediatelimit}
\begin{align}
		\dot{S_2} &= (-S_2+\lambda_2)D-y_3\mu_2(S_2,S_3)X_2,\label{eq:intermediateaceticacid}\\
		\dot{S_3} &= -DS_3+\lambda_3D,\label{eq:intermediateammonia}\\
		\dot{X_2} &= -D_2X_2 + \mu_2(S_2,S_3)X_2.\label{eq:intermediatemethanogens}
\end{align}
\end{subequations}
		Since we are assuming  that $\mu_1(S_1)$ is a monotone
		response function,  the results in \cite{WL:1992}
		can be applied to
		the first two equations in
		\cref{eq:system} to prove that
		 $(S_1(t), X_1(t))$ converge exponentially to
		$(\lambda_1, X_1^*)$ as $t\to \infty$.  (The restriction
		that the results in \cite{WL:1992} only apply to a general class of monotone response functions rather than any monotone
		response function  does not apply to the single
		species growth model.)
	
		Let $x(t) = (S_1(t),X_1(t),S_2(t),S_3(t),X_2(t))$ be
		any solution of \cref{eq:system},   $K\subset \mathbb{R}^5_+$ be a compact
		set, and let  $||\cdot||$ denote the Euclidean norm. For $t_0\geq 0$, consider
\begin{align*}
 \mathcal{Q}_1 = \int_{t_0}^\infty \sup_{x\in
	K}||(Y_1(t),Y_2(t),Y_3(t),Y_4(t),0)|| \ dt,
\end{align*}
where	
\begin{align*}
	Y_1(t) &= (S^{(0)}-S_1(t))D-y_1\mu_{1}(S_1(t))X_1(t),\\
	Y_2(t) & = -D_1X_1(t)+\mu_{1}(S_1(t))X_1(t),\\
	Y_3(t) & = y_2\mu_1(S_1(t))X_1(t)-D\lambda_2,\\
	Y_4(t) & = y_4\mu_1(S_1(t))X_1(t)-D\lambda_3.
\end{align*}
  If $t_0 = 0$, then for any $0<t_1<\infty$, by continuity of the norm,
\begin{equation}
	\int_0^{t_1}\sup_{x\in K}||(Y_1(t),Y_2(t),Y_3(t),Y_4(t),0)|| \
	dt <\infty.
\end{equation}
 Thus, we need only consider the case $t_0>0$. By the Cauchy-Schwartz inequality,
\begin{align}\label{eq:Quasi-bound}
&\mathcal{Q}_1 \leq\left(\int_{t_0}^\infty \frac{1}{t^2}\, dt\right)^{\frac{1}{2}}\Bigg(\int_{t_0}^\infty t^2\sup_{x\in K}(Y_1(t)^2+Y_2(t)^2+Y_3(t)^2+Y_4(t)^2)\, dt\Bigg)^{\frac{1}{2}}.
\end{align}
The first integral, $\int_{t_0}^\infty \frac{1}{t^2}\, dt$,
is finite. Since all of the terms of
\begin{equation}\label{eq:Quasi-bound-parts}
\int_{t_0}^\infty t^2\sup_{x\in
	K}\left(Y_1(t)^2+Y_2(t)^2+Y_3(t)^2+Y_4(t)^2\right)\  dt,
\end{equation}
		are positive, we   can consider them individually. We begin with the second term,
\begin{align*}
	\int_{t_0}^\infty t^2\sup_{x\in K}Y_2(t)^2dt=  \int_{t_0}^\infty t^2 \sup_{x\in K} X_1^2(t)\left[-D_1 +\mu_1\left(S_1(t)\right)\right]^2dt.
\end{align*}
		Since $\mu_1(S_1) \in C^1$, by the Mean Value Theorem, for every $t>0$, there
	exists $\theta(t)$, such that $S_1(\theta(t))$ lies between
	$S_1(t)$ and  $\lambda_1$. Let $M_0=\sup_{t\in [0, \infty)}
		|\mu_1'(S_1(\theta(t))|  > 0$. Since $S_1(t)\to \lambda$ as $t\to
		\infty$, $\mu_1'(S_1(\theta(t))$ remains bounded, $M_0$
		is finite and
	$| - D_1 + \mu_1(S_1(t))|=
	| -\mu_1(\lambda_1) + \mu_1(S_1(t))|
	 =
		|\mu_1'(S_1(\theta(t)))| | -\lambda_1+ S_1(t)|\leq M_0 |
		-\lambda_1+ S_1(t)| \to 0,$  exponentially as $t\to 0$.
		 Thus, there is a $k>0$, such that
\begin{align*}
& \int_{t_0}^\infty t^2\sup_{x\in K}X_1^2(t) [-D_1 +\mu_1(S_1(t))]^2dt \leq   \overline{X_1} \widetilde{M_0}^2\int_{t_0}^\infty t^2e^{-2kt}dt <\infty,
\end{align*}
where $\overline{X_1}$ is the maximum value of $X_1(t) \in K$, and $\widetilde{M_0}=M_0|S_1(0)-\lambda_1|$.

We now consider the first term,
\begin{align*}
\int_{t_0}^\infty t^2\sup_{x\in K}Y_1&(t)dt =\int_{t_0}^\infty t^2 \sup_{x\in K} \big[ (S^{(0)}-S_1(t))D
-y_1\mu_{1}(S_1(t))X_1(t) \big]^2dt \\
&\leq \int_{t_0}^\infty t^2\sup_{x\in K} \big[(\lambda_1-S_1(t))D-y_1\mu_1(S_1(t))X_1(t)+ (S^{(0)}-\lambda_1)D\big]^2dt.
\end{align*}
By Young's inequality ( \textit{i.e.} that for any two real numbers, $a$ and $b$, $(a+b)^2\leq 2a^2+2b^2$) and using $D(S^{(0)}-\lambda_1) = D_1y_1 X_1^*$,
\begin{align*}
\int_{t_0}^\infty t^2 \sup_{x\in K} Y_1(t)dt
{\leq}\int_{t_0}^\infty 2t^2 \sup_{x\in K} \big[(\lambda_1-S_1(t))^2D^2{+}y_1^2\left(\mu_1(S_1(t))X_1(t) - X_1^*D_1\right)^2 \big]dt.
\end{align*}
Since this integral is a sum of positive terms we may consider each term individually. The first term is bounded above by the integral of a decaying exponential, and so is finite. We use Young's inequality to bound the second term,
\begin{align}
&2\int_{t_0}^\infty t^2 \sup_{x\in K} \Big[ y_1^2\left(\mu_1(S_1(t))X_1(t)-D_1X_1(t)+D_1X_1(t)- X_1^*D\right)^2\Big]dt\notag\\\label{eq:Ineq2}
&\leq 4y_1^2\int_{t_0}^\infty t^2 \sup_{x\in K}\Big[ X_1(t)^2(\mu_1(S_1(t))-D_1)^2+D_1^2\left(X_1(t) -X_1^*\right)^2 \Big]dt,
\end{align}
where both of the terms in \cref{eq:Ineq2} are bounded above by a decaying exponential and so this integral is finite.
For the third term in \cref{eq:Quasi-bound-parts}, write
\begin{align*}
&\int_{t_0}^\infty t^2\sup_{x\in K}\left[y_2\mu_{1}(S_1(t))X_1(t)-D\lambda_2\right]^2dt\\&= \int_{t_0}^\infty t^2y_2^2\sup_{x\in K}\left[\mu_{1}(S_1(t))X_1(t) -D_1X_1(t)+ D_1X_1(t) - \frac{D \lambda_2}{y_2}\right]^2 dt \\
&\leq 2\int_{t_0}^\infty t^2y_2^2\sup_{x\in K}\left[\mu_{1}(S_1(t))X_1(t){-}D_1X_1(t)\right]^2dt{+}2{\int_{t_0}^\infty} {t^2}\sup_{x\in K}\left[y_2D_1X_1(t) {-} D\lambda_2 \right]^2dt.
\end{align*}
Noting that $y_2 D_1 X_1^* = D \lambda_2$, the exponential decay of
$(X_1(t)-X_1^*)^2$, and the same decay arguments as with the first term
in \cref{eq:Quasi-bound-parts}. The finiteness of the fourth term in
\cref{eq:Quasi-bound-parts} follows from a similar idea, noting that
$y_4 D_1 X_1^* = D \lambda_3$. Thus, \cref{eq:system} is
quasi-autonomous with limiting system \cref{eq:intermediatelimit}.

 Now we finally show that \cref{eq:system} has limiting system
\cref{eq:limitsystem}.  From  \cref{eq:intermediateammonia}, if follows that
\begin{equation}\label{eq:S3bound}
	|S_3(t)- \lambda_3| = \left|S_3(0)-\lambda_3 \right|e^{-Dt}.
\end{equation}
We use this to argue that
\begin{equation*}
\mathcal{Q}_2 = \int_{t_0}^\infty \sup_{x\in K} \sqrt{(y_3^2+1)Y_5(t)^2 + D^2Y_6(t)^2}dt< \infty,
\end{equation*}
where, $Y_5(t) = \mu_2(S_2(t),\lambda_3)-\mu_2(S_2(t),S_3(t))X_2(t)$, and $Y_6(t) =S_3(t) - \lambda_3$. The Cauchy-Schwartz inequality allows us to split the integral into more manageable pieces,

\begin{align}
\mathcal{Q}_2&\leq \left(\int_{t_0}^\infty\frac{1}{t^2}dt\right)^{\frac{1}{2}}\left(\int_{t_0}^\infty t^2\sup_{x\in K}[(y_2^2+1)Y_5(t)^2 + D^2Y_6(t)^2]dt\right)^\frac{1}{2}.\notag
\end{align}
By \cref{eq:S3bound}, the term containing $Y_6(t)$ is bounded above. In order to show the integral containing $Y_5(t)$ is bounded above we use the fact that $\mu_2(S_2,S_3)\in C^1$ and \cref{eq:S3bound} to argue that there exists $M_1\geq0$ such that
\begin{align*}
|\mu_2(S_2(t),\lambda_3)-\mu_2(S_2(t),S_3(t))| \leq M_1|S_3(0)-\lambda_3|e^{-Dt}.
\end{align*}
Since $X_2(t)$ is bounded we have
\begin{equation}\label{eq:Y5bound}
\int_{t_0}^\infty t^2\sup_{x\in K} \Big[ (y_3^2+1)\overline{X_2}M_1|S_3(0)-\lambda_3|e^{-Dt}\Big] dt,
\end{equation}
Where $\overline{X_2}$ is the maximum value of $X(t)$ in $K$. The
integral on the right is finite and therefore, by \cref{lem:quasi-autonomous}, \cref{eq:system} is quasi-autonomous with
limiting system \cref{eq:limitsystem}.
\end{proof}

\end{appendix}
\bibliographystyle{siamplain}

\bibliography{SIAPM119878Bib}

\end{document}